\documentclass[10pt,journal,compsoc]{IEEEtran}
\usepackage{ragged2e}
\usepackage{amsmath}

\usepackage{amsthm}
\usepackage{bbm}
\usepackage{bm}
\usepackage{algorithm}
\usepackage{algorithmic}
\usepackage{multirow}
\usepackage{booktabs}
\usepackage{tabularx}
\usepackage{color}
\usepackage{colortbl}
\usepackage{enumerate}
\usepackage{graphicx}
\usepackage{subfig}
\usepackage[colorlinks,linkcolor=blue,anchorcolor=blue,citecolor=blue]{hyperref}
\usepackage{amssymb}
\usepackage{enumitem}
\usepackage[justification=centering]{caption}
\usepackage{color,xcolor}
\usepackage{setspace}  
\definecolor{shadecolor}{rgb}{0.8,1.0,1.0}  
\usepackage{framed} 
\usepackage[framemethod=tikz]{mdframed}

\newtheorem{Example}{Example}
\newtheorem{lemma}{Lemma}
\newtheorem{theorem}{Theorem}
\newtheorem{remark}{Remark}
\newtheorem{assumption}{Assumption}

\newcommand{\tabincell}[2]{\begin{tabular}{@{}#1@{}}#2\end{tabular}}

\definecolor{curve_pink}{HTML}{9b59b6} 
\definecolor{curve_blue}{HTML}{3498db} 

\ifCLASSINFOpdf
\else
\fi
\begin{document}

\title{Meta-Wrapper: Differentiable Wrapping Operator for User Interest Selection in CTR Prediction}
%
%
%

\author{Tianwei Cao,
	Qianqian Xu\IEEEauthorrefmark{1},~\IEEEmembership{Senior Member,~IEEE,}
	Zhiyong Yang,
	and Qingming Huang\IEEEauthorrefmark{1},~\IEEEmembership{Fellow,~IEEE}
	\IEEEcompsocitemizethanks{
		\IEEEcompsocthanksitem Tianwei Cao is with the School of Computer Science and Technology,
		University of Chinese Academy of Sciences, Beijing 101408, China (email: \texttt{caotianwei19@mails.ucas.ac.cn}).\protect\\
		\IEEEcompsocthanksitem Qianqian Xu is with the Key Laboratory of
		Intelligent Information Processing, Institute of Computing Technology, Chinese
		Academy of Sciences, Beijing 100190, China, (email: \texttt{xuqianqian@ict.ac.cn}).\protect\\
		\IEEEcompsocthanksitem Zhiyong Yang is with the School of Computer Science and Technology,
		University of Chinese Academy of Sciences, Beijing 101408, China
		(email: \texttt{yangzhiyong@iie.ac.cn}). \protect\\
		\IEEEcompsocthanksitem Qingming Huang is with the School of Computer Science and Technology,
		University of Chinese Academy of Sciences, Beijing 101408, China, also
		with the Key Laboratory of Big Data Mining and Knowledge Management (BDKM),
		University of Chinese Academy of Sciences, Beijing 101408, China,  also
		with the Key Laboratory of Intelligent Information Processing, Institute of
		Computing Technology, Chinese Academy of Sciences, Beijing 100190, China, and also with Peng Cheng Laboratory, Shenzhen 518055, China
		(e-mail: \texttt{qmhuang@ucas.ac.cn}).\protect\\
		\IEEEcompsocthanksitem * corresponding author \protect\\
}}

%
%

\markboth{TO APPEAR IN IEEE TRANSACTIONS ON PATTERN ANALYSIS AND MACHINE INTELLIGENCE}%
{Shell \MakeLowercase{\textit{et al.}}: Bare Demo of IEEEtran.cls for Journals}
%



\maketitle

\begin{abstract}
\justifying
Click-through rate (CTR) prediction, whose goal is to predict the probability 
of the user to click on an item, has become increasingly significant in 
the recommender systems. 
Recently, some deep learning models with the ability to automatically 
extract the user interest from his/her behaviors have achieved great success. 
In these work, the attention mechanism is used to select the user 
interested items in historical behaviors, improving the performance of 
the CTR predictor. Normally, these attentive modules can be jointly trained 
with the base predictor by using gradient descents. In this paper, we 
regard user interest modeling as a feature selection problem, which 
we call user interest selection. For such a problem, we propose a novel 
approach under the framework of the wrapper method, which is named  
Meta-Wrapper. More specifically, we use a differentiable module as our 
wrapping operator and then recast its learning problem as a continuous 
bilevel optimization. Moreover, we use a meta-learning algorithm to solve the 
optimization and theoretically prove its convergence.
Meanwhile, we also provide theoretical analysis to show that 
our proposed method 1) efficiencies the wrapper-based feature selection, and
2) achieves better resistance to overfitting.
Finally, extensive experiments on three public datasets manifest 
the superiority of our method in boosting the performance of CTR prediction.

\end{abstract}

%
\IEEEpeerreviewmaketitle

\section{Introduction}
%
%
%
%
\IEEEPARstart{C}lick-through rate (CTR) prediction plays a centric role in 
recommender systems, where the goal is to predict the probability that 
a given user clicks on an item. In these systems, each recommendation 
returns a list of items that a user might prefer in terms of the prediction. 
In this way, the performance of the CTR prediction model is directly related 
to the user experience and thus has a critical influence on the final revenue 
of an online platform. Therefore, this task has attracted a large number of 
researchers from the machine learning and data mining community.

Currently, there are many CTR prediction methods 
\cite{PNN, DFM, XDFM} that 
focus on automatic feature engineering. The main idea behind 
these approaches is learning to combine the features automatically 
for better representation of instances. Following this trend, 
recent researches pay particular attention to the interplay of the target 
item and user behavioral features. 
Deep Interest Network (DIN) \cite{DIN} 
is one of the typical examples. Given a specific user, DIN simply uses an 
attention layer to locally activate the historically clicked items that 
are most relevant to the target item. Here the outputs of the attention 
layer come from the interaction between the given target and each clicked 
item, characterizing the interests of the user. Similar attentive modules 
are also used in \cite{DIEN, DSIN, HPMN, MIMN} ,etc. All of these methods 
capture user interests by means of historically clicked items, 
under the assumption that user's preference is covered by one's behaviors.

Actually, such an attention mechanism can be regarded as a special case of 
feature selection methods. More specifically, it aims to find the most relevant 
features (clicked items) according to a given target. 
Its main difference with traditional feature selection is the 
membership of features in the subset is not necessarily assessed in binary terms. 
Here the feature subset could be a fuzzy set and the attention unit 
can be regarded as the membership function. 

However, there exists a fact that is not paid enough attention, namely, 
feature selection based on attention mechanisms also increase overfitting 
risks. More specifically, when we regard the attentive module 
as a feature selector, the other parts of the model are naturally 
viewed as a base CTR predictor, where the two components are 
jointly learned from the same training set. Based on this, we can 
expect that the model with a feature selector is more capable of 
fitting the training data than using the base predictor only. 
However, such an approach also introduces additional complexity 
to the architecture of the predictive model, which implies a higher 
risk of overfitting. In other words, there is still an improvement 
space in the generalization performance of those feature selection 
methods based on attention mechanisms.

Motivated by this fact, we argue that the feature selector should be learned 
not only from the same data with base predictor, but also from out-of-bag data.  
In this way, 
the knowledge learned in feature selector would come from different sources, 
thus alleviating the possibility of overfitting on a single dataset. 
Such an idea is exactly inspired by wrapper method 
\cite{fswrapper, fswrapper2010}, which is a classical approach to 
feature selection. This kind of methods regard feature selection 
as a non-continuous bilevel optimization (BLO) 
\cite{bilevel_meta} problem which is shown in Fig.(\ref{fig:main_a}). 
Instead of solely relying on the training set,
the wrapper method takes the performance on the out-of-bag data as 
the metric of the feature selection. 
In this way,
two datasets with different empirical distribution are utilized during 
the learning process, 
thereby mitigating the risk of overfitting.

\begin{figure*}[htbp]
	\centering
	
	\subfloat[Traditional Wrapper Method]{
		\label{fig:main_a}
		\includegraphics[scale=0.45, trim=200 80 200 77, clip]{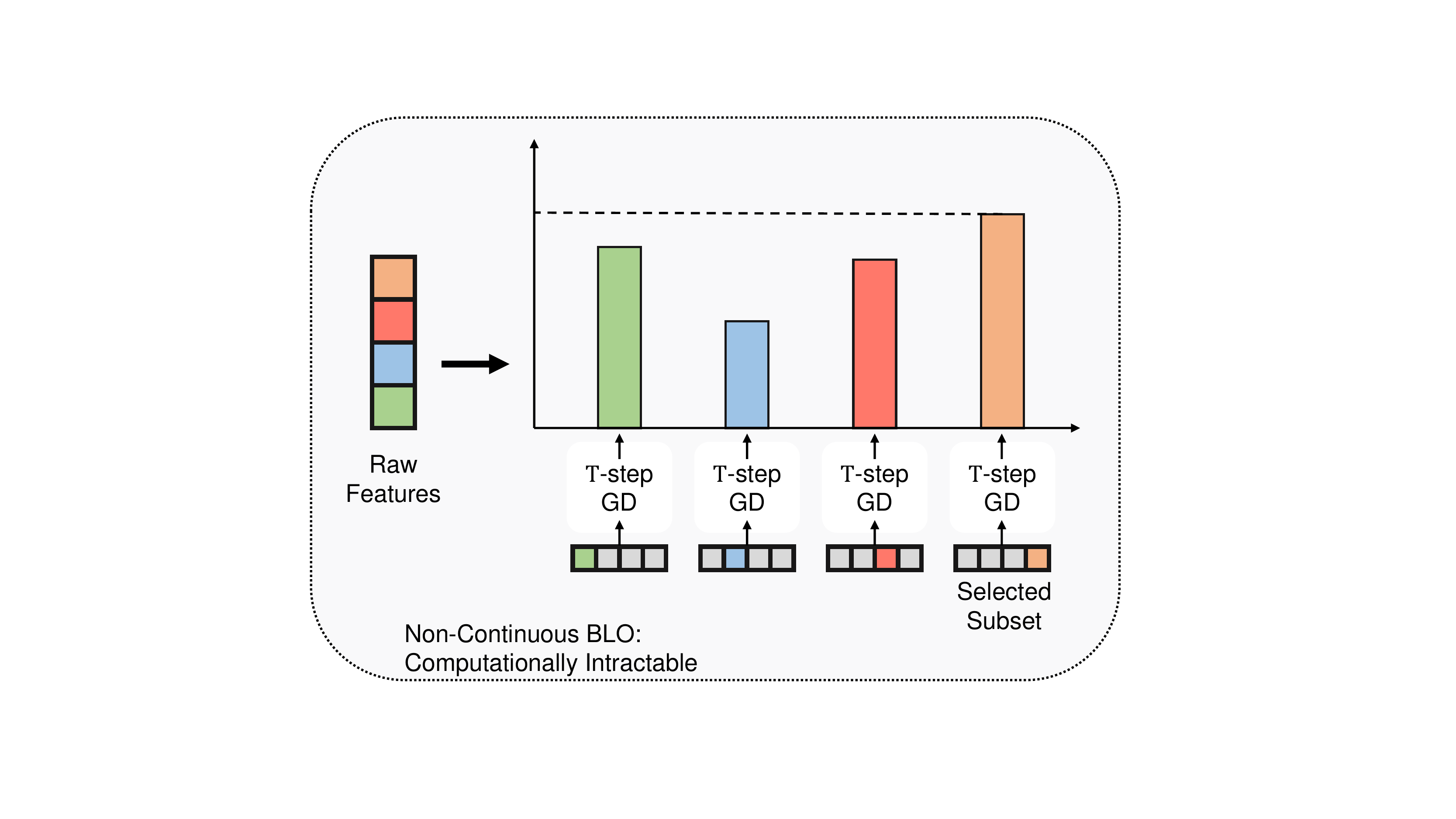}} 
	\subfloat[Our Proposed Meta-Wrapper]{
		\label{fig:main_b}
		\includegraphics[scale=0.45, trim=200 80 200 77, clip]{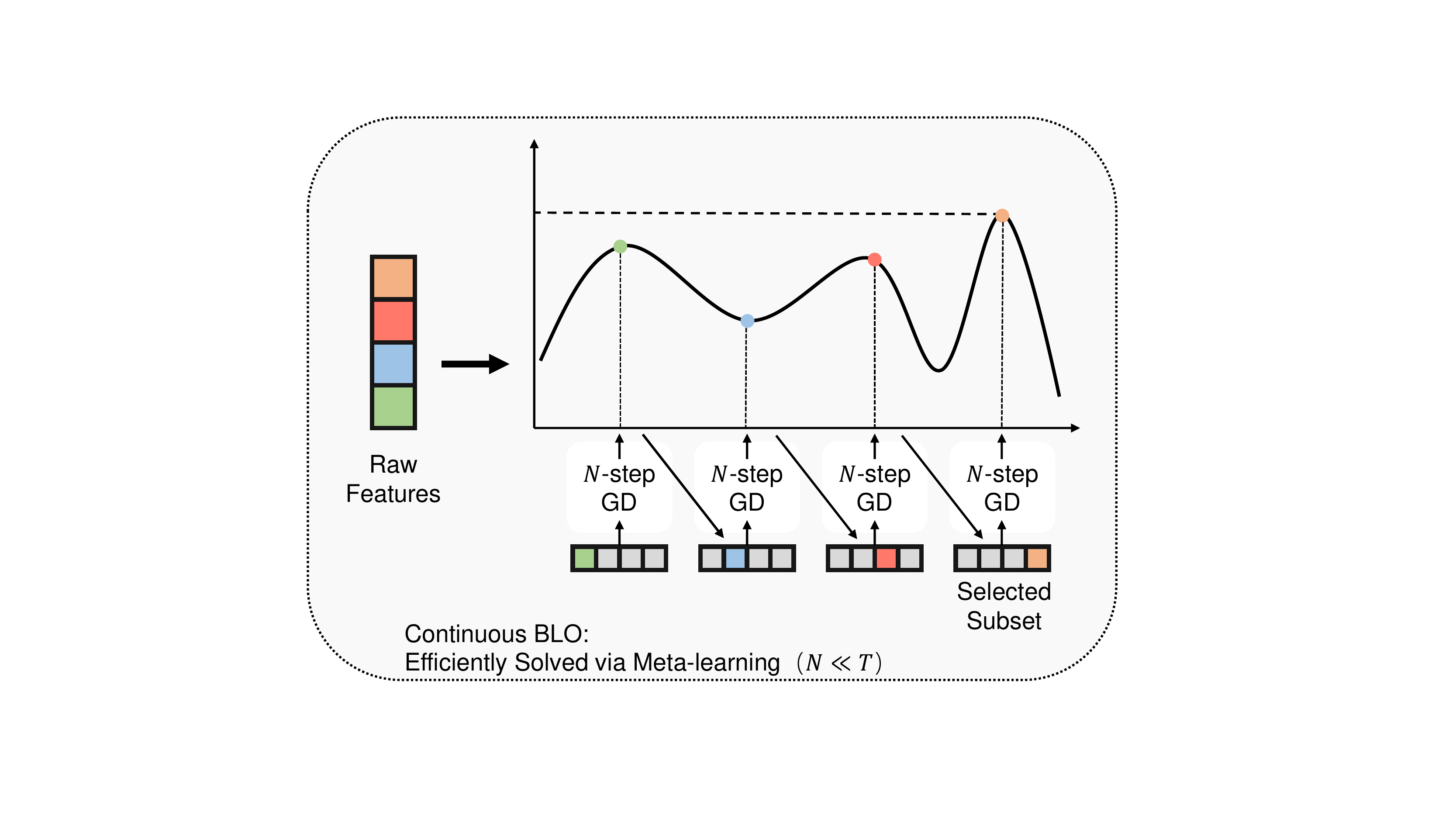}} 
	\centering
	\caption{Comparison between Traditional Wrapper Method and Meta-Wrapper.}
	\label{fig:main}
\end{figure*}

However, there are two significant problems left to be solved 
in traditional wrapper methods:
\begin{itemize}
	\item \textbf{Intractable computational complexity}:
	The number of possible feature combinations 
	increases exponentially with feature dimensions. 
	Therefore, it is actually an NP-hard problem to find 
	the best feature subset. To get an acceptable selection 
	result, the model must be trained again and again 
	on different feature subsets, which leads to a great computing cost.
	\item \textbf{Insufficient flexibility}:
	In CTR prediction, 
	the selection of the features (clicked items) 
	should rely on the target item, which cannot be 
	directly implemented via traditional wrapper methods. 
	Meanwhile, wrapper methods search the best feature combination 
	in discrete space instead of continuous space. Thus it cannot 
	be used as a differentiable sub-module, like attention mechanisms, 
	of an end-to-end deep learning model.
\end{itemize}
In order to solve these problems, we propose a novel feature selection approach 
under the framework of the wrapper methods, which is called Meta-Wrapper. 
Our main idea is to encapsulate the inner-level learning process
by means of a differentiable wrapping operator. 
By applying this operator, the wrapper-based feature selection is 
recast to a continuous BLO problem that can be efficiently solved through
meta-learning \cite{META_LEARNING_REVIEW}. 
To this end, we, following the similar spirit of attention methods, 
regard the feature selection problem as learning a scoring function. 
This function takes behavioral features and the target item as inputs, 
then outputs the relevant 
scores used for properly selecting features. 
More specifically, the novelty of our proposed method can be viewed from two perspectives:
\begin{itemize}
	\item 
	As is shown in Fig.(\ref{fig:main}), 
	our feature selector could be meta-learned in a more efficient manner
	compared with the traditional wrapper method.
	We will elaborate on this meta-learning algorithm and its convergence 
	result in the following sections.
	\item Compared with attention mechanisms, 
	our feature selector is learned from two different sources. 
	This actually leads to an implicit regularization term to prevent overfitting.
	In later sections, we also provide theoretical analysis on this point. 
\end{itemize}
In other words, our work bridges the gap between deep-learning-based 
feature selector and classical wrapper method, in the sense that 
the generalization ability of the feature selector can be improved through an
efficiently meta-learned wrapping operator.

To summarize, the main contributions of this paper are three-fold:
\begin{itemize}
	\item We propose a novel approach to user interest 
	selection, which encapsulates the optimizing iterations
	by means of a differentiable wrapping operator. 
	Such an approach unifies feature
	selection and CTR prediction into a continuous BLO objective.
	\item From theoretical analysis, we show
	that there exists an implicit regularization term in such a
	BLO objective, which explains how this objective 
	alleviates overfitting. 
	\item We develop a meta-learning algorithm for solving this BLO problem
	in an efficient way. 
	Taking a step further, we theoretically prove the convergence of this algorithm
	under a set of broad assumptions.
\end{itemize}

\section{Related Work}


\subsection{CTR prediction}


As an early trial, researchers solve the CTR prediction problem mainly through 
simple models with expert feature engineering. \cite{trenches}
integrates multiple discrete features and trains a linear predictor
through an online sparse learning algorithm. 
\cite{feedback} introduces a machine learning framework 
based on Logistic Regression, aiming to predict CTR for 
online ads. \cite{ensemble} ensembles collaborative filters, probability models 
and feature-engineered models to estimate CTR. \cite{FB} uses 
ensembled decision trees to capture the features of social ads,
thereby improving the model performance. In spite of the remarkable success
of these models, they heavily rely on manual feature engineering, 
which results in significant labor costs for developers. 


To address this problem, Factorization Machines (FM) \cite{FM} are 
proposed to automatically perform the second-order feature-crossing
by inner production of feature embeddings. 
More recently, 
\cite{CFM, FAFM,AFM,FWFM,Roubst_FM,FM_aaai_2019_1,FM_aaai_2019_2,FM_aaai_2019_3} 
improve the FM model from different perspectives, 
for the sake of better predictive performance. 
Notably, such embedding mechanisms
have better resistance to the sparsity of crossed features 
since the embedding for each individual feature is trained independently. 
To capture more sophisticated feature interactions, 
there is a line of works using the universal approximation ability 
of Deep Neural Networks (DNN) to improve the feature embedding. 
Deep crossing model \cite{DCROSS} feeds the feature embeddings 
into a cascade of residual units \cite{resnet}. 
Moreover, Factorization-machine supported Neural Network (FNN) \cite{FNN} 
stacks multiple dense layers on top of an embedding layer transferred from 
a pre-trained FM. Based on a similar architecture, 
Product-based Neural Network (PNN) \cite{PNN} adds a product 
layer between the embedding layer and the dense layers. 
Following a similar spirit, Neural Factorization Machine (NFM) \cite{NFM} 
adds a bi-interaction layer 
before the dense layers to enhance the embeddings, thereby achieving 
state-of-the-art performance via fewer layers. 
It is equivalent to stack dense layers on a second-order FM.
More recently, \cite{AUTOINT} automatically 
learns the high-order feature interactions 
by multi-head self-attention \cite{Transformer}, 
which is called Automatic Feature Interaction Learning (AutoInt). 
By leveraging the ability of representation learning, 
these models implement automatic feature-crossing in CTR prediction.
There is also another line of works that consider low- and high-order 
feature-crossing simultaneously. 
The well-known Wide\&Deep model \cite{WD} captures low- and high-order 
feature-crossing 
in terms of wide and deep components, respectively.   
After this work, there emerge some related researches that 
improve the wide or deep component. 
DeepFM \cite{DFM} uses second-order FM as the wide component to 
automatically capture low-order feature-crossing.
\cite{FPENN} proposes a model named Field-aware 
Probabilistic Embedding Neural Network (FPENN).
It uses field-aware probabilistic embedding (FPE) strategy
to enhance the wide and deep components simultaneously.
Moreover, Deep \& Cross Network (DCN) \cite{DCN} utilizes a
cross network as the wide component, for modeling the N-order 
feature-crossing in a more efficient way. 
\cite{XDFM} proposes eXtreme Deep Factorization Machine (xDeepFM). 
It uses a Compressed Interaction Network (CIN) as the wide component, 
thereby improving the representative ability compared with DCN.
Taking a step further, \cite{He_plus_2} uses graph neural networks (GNN)
with adversarial training strategy
to capture more complex feature relations.
This method successfully achieves robustness towards the perturbations. 
Additionally, \cite{He_plus_1} proposes a novel method to 
capture interest features for cold-start users, which achieves 
state-of-the-art performance in this setting.
All of these methods clearly show the effectiveness of automatic feature 
engineering in this area.

Different from the previous research, recent methods
pay more attention to model users' interest by using behavioral features. 
Notably, feature interaction and user interest modeling are not contradictory 
to each other. 
The latter is mainly an approach to deal with behavior feature sequences.
Deep Interest Network (DIN) \cite{DIN} uses attention mechanisms to adaptively 
activate user interests for a given item, where the user interests are modeled
as his/her clicked items.
This method adaptively produces the user interest representation by selecting the 
most relevant behaviors to the target item. 
In other words, the user representation varies over different
targets, thus improving the expressive ability. 
Deep Interest Evolution Network (DIEN) \cite{DIEN} 
designs an RNN-based module as an interest extractor to 
model the user interest evolution. 
Deep Session Interest Network (DSIN) \cite{DSIN} splits the users' 
behavioral sequences as multiple sessions,
for modeling their interest evolution at session-wise level.
\cite{HPMN} proposes a model named Hierarchical Periodic Memory Network (HPMN). 
This model uses memory networks to incrementally capture 
dynamic user interests. 
Similarly, Multi-channel user Interest Memory Network (MIMN) \cite{MIMN}
also takes the memory network to capture the long-term user 
interests and designs a separate module named User Interest Center (UIC) 
to reduce the latency of online serving. 
All these models are developed from the basis of DIN, 
improving the user interest modeling from the time series perspective.

Following these works, we also focus on capturing the interests in the user 
behavioral features. 
However, we recast the user modeling as a feature selection problem 
rather than capturing the user interest evolution, 
which provides a novel perspective for the CTR prediction task.

\subsection{Feature Selection}

Feature selection methods can be normally divided into three categories: 
filter, wrapper, and embedded methods \cite{fs2003, fs2017}. 
In this section, we will discuss the three kinds of methods respectively.
Filter methods rank the features by their related scores to the target. 
Different filter approaches adopt various kinds of metrics 
to score the relevance, such as T-score \cite{TSCORE}, 
information-theory based criteria \cite{INFOFS}, 
similarity metrics \cite{SPEC}, and fisher-score \cite{FISHER} .etc. 
These methods can only select the relevant features 
by some pre-defined rules, leading to the inability to adaptively modulate
selected results.
Wrapper methods \cite{fswrapper2010} compare validating performances of 
the models with different feature subsets. 
To this end, the model must be repeatedly trained for every possible 
subset, resulting in intractable computational complexity \cite{fswrapper}. 
In this way, wrapper methods are not desirable for the tasks with large-scale
datasets, such as recommender systems.
In the embedded methods, a feature selector is jointly learned with 
the predictive model, which has more efficiency than the wrapper methods. 
In an early trial, there exist many approaches that penalize weights of irrelevant 
features via specific regularization terms.
\cite{l21fs} uses $L_{2,1}$-norm to learn sparse feature weights, 
where the weights of irrelevant features are reduced to zero.
for the sake of better flexibility, \cite{PL21} uses the probabilistic $L_{2,1}$-norm 
to select features, where feature subset can be different for various targets. 
With the popularity of deep learning, there emerge
another kind of embedded methods based on neural networks, 
which are called attention mechanisms \cite{NMT}. 
This strategy is widely applied in many areas, such as 
object detection \cite{RichFS, SOD, sptransformer, qimg}, 
machine translation \cite{NMT, Transformer}, 
video description \cite{DVLSTM} and
diagnosis prediction \cite{Dipole}.
All of these methods use attentive modules to produce 
real-valued weights to capture feature importance.  
To obtain discrete feature subsets, 
\cite{DFS} and \cite{FSD} design special attentive modules
which provide sparse weights. 
However, the sparsity of the selected subset, for recommender systems, 
is actually less important than its representative ability.


In this paper, we follow the spirit of the wrapper and embedded method
simultaneously. We model the feature selector as a neural network 
meta-learned jointly with base predictor, 
thereby avoiding excessive computing costs. 
Meanwhile, flexibility can also be 
taken into account to handle the diversity of user interests.

\section{Problem Statement}

We commence by introducing deep-learning based CTR prediction. Then we 
formulate the user interest selection problem in this field. Notations
defined in this section will be adopted in the entire paper.
More notations and their descriptions are listed
in Appendix.

\subsection{Deep-Learning Based CTR Prediction}
To define the user interest selection, we first need to formulate the CTR 
prediction task. Given a set of triples $\{({u}, {v}, {y})\}$ 
each of which includes the user $u\in\mathcal{U}$, item $v\in\mathcal{V}$ and the 
corresponding label of user click indicator
$$
{y}=\left\{
\begin{array}{lr}
	1,\ \ {u}\ has\ clicked\ {v}; \\
	0,\ \ otherwise,              \\
\end{array}
\right.
$$
our goal is to predict click-through rate $\hat{{p}}
({y}=1|{u},{v})$. 
Here the $({u},{v})$ are indices of the user and target item, while
other inputs are omitted for their non-directly relevance to our problem.
We implement CTR prediction through a 
learned function $f_{\bm{\theta}}(\cdot)$ 
with parameters $\bm{\theta}$, which can be simply formulated as
$$
\hat{{p}}({y}=1|{u},{v})
=f_{\bm{\theta}}({u},{v}).
$$
To learn $f_{\bm{\theta}}(\cdot)$, the cross-entropy loss 
$L({y},\hat{{p}})$ is used as the optimization target
\begin{equation}\label{cross_entropy}
	L({y},\hat{{p}})=
	-{y}\log(\hat{{p}}) - (1-{y})\log(1-\hat{{p}}).
\end{equation}

Consider that $f_{\bm{\theta}}(\cdot)$ is usually modeled as a deep 
neural network and $L({y},\hat{{p}})$ would be optimized 
by the algorithms with gradient operations. 
The categorical inputs of $f_{\bm{\theta}}(\cdot)$ 
have to be transformed into real-value vectors through an embedding 
layer. More formally, the predictive model can be further represented as
\begin{equation}\label{ctr_with_emb}
	\hat{{p}}({y}=1|{u},{v})=f_{
		\bm{\theta}}(\bm{\theta}^{E}_{u},
	\bm{\theta}^{E}_{v}),
\end{equation}
where $\bm{\theta}^{E}\in\mathbb{R}^{{K} \times {M}}$ is the 
embedding matrix; ${M}$ is the number of categorical features in all; 
${K}$ is the dimension of each embedding vector; 
$(\bm{\theta}^{E}_{u},\bm{\theta}^{E}_{v})$ are embeddings
of the user ${u}$ and target item ${v}$, respectively.

\subsection{Exploiting User Behavioral Features}
From Eq.(\ref{ctr_with_emb}), we observe that traditional CTR prediction methods 
simply represent each user as a ${K}$-dimensional embedding 
vector. However, recent works \cite{svd_pp} find such a method seems 
insufficient for learning user preference
in real-world recommender systems. 

for the sake of better representation, augmenting the user with his/her 
historical features \cite{DIN, NCF, OPNCF} has become a widely adopted paradigm. 
More Specifically, historical features of a user naturally refer to the sequence of 
his/her clicked items. Based on this, we can represent each user ${u}$ as 
\begin{equation}\label{user_repr}
	\bm{r_{{u}}} = 
	[\bm{\theta}^{E}_{v_1}; 
		\bm{\theta}^{E}_{v_2};
		\cdots
	\bm{\theta}^{E}_{v_{T_u}}],
\end{equation}
where $\bm{r_{{u}}}\in\mathbb{R}^{{K}\times {T_u}}$ represents
the user's behavioral sequence; 
${T_u}\in\mathbb{N}^{+}$ is the number of clicked items; 
${v_1}\dots{v_{T_u}}$ are clicked items; 
each $\bm{\theta}^{E}_{v_{.}}\in\mathbb{R}^{{K}}$ is the embedding of 
a clicked item.

In this way, a CTR prediction model with user behavioral features can be formulated as 
\begin{equation}\label{ctr_with_user}
	\hat{{p}}({y}=1|{u},{v})=f_{
		\bm{\theta}}(\texttt{Pooling}(\bm{r_{{u}}}),
	\bm{\theta}^{E}_{v}), 
\end{equation}
where $\texttt{Pooling}(\bm{r_{{u}}})
\in\mathbb{R}^{{K}}$ is the result of applying a pooling function 
to user representation. By taking the behavioral features into account, 
we replace $\bm{\theta}^{E}_{u}$ in Eq.(\ref{ctr_with_emb})
with a more comprehensive user representation $\bm{r_{{u}}}$. 

\subsection{Modeling User Interest by Feature Selection}
Despite that introducing behavioral features could enhance user representation, there is a
fact that the number of clicked items is continually growing over time. With this trend, 
some user's behavioral sequence would become extremely long. In this scene, 
not all historical behaviors can play a positive role in every specific prediction.
For example, when  a user shopping for clothes on the e-commerce platform, 
it is totally irrelevant to which books he/she bought in the past.

Therefore, we need to filter out the non-trivial items from each user's 
behavioral sequence, keeping the items that are beneficial to downstream prediction. 
Since the remaining items are considered more representative for user interest, we name such
a process as \textbf{user interest selection}. Essentially, it can be regarded as a 
special case of feature selection problems. 

More formally, our goal is to learn a function $g_{\bm{\phi}}(\cdot)$ 
with parameter $\bm{\phi}$ for selecting user interests:
\begin{equation}\label{seletor}
	\bm{s_{{u}}} = g_{\bm{\phi}}(
	\{\bm{\theta}^{E}_{v_i}\}_{{i}=1}^{T_u}
	, \bm{\theta}^{E}_{v}),
\end{equation}
where $\bm{\theta}^{E}_{v}$ is the embedding of the target item; 
$\bm{s_{{u}}}\in\mathbb{R}^{{T_u}}$ represents the 
relevant scores between each behavioral feature and the 
target ${v}$:
$$
\left\{
\begin{array}{lr}
	\bm{s_{{u}}}[i] > \bm{s_{{u}}}[j], 
	\ \ {v_i}\ is\ more\ relevent\ to\ {v};    \\
	\bm{s_{{u}}}[i] < \bm{s_{{u}}}[j], 
	\ \ {v_j}\ is\ more\ relevent\ to\ {v};    \\
	\bm{s_{{u}}}[i] = \bm{s_{{u}}}[j], 
	\ \ otherwise.                                           \\
\end{array}
\right.
$$ 
Similar to \cite{NAIS} and \cite{DIN}, we take the inner product
$$
\bm{\hat{r}_{{u}}}=
(\bm{r_{{u}}} \bm{s_{{u}}})\in\mathbb{R^{{K}}}
$$
as the aforementioned pooling function. 
In other words, it is a weighted sum of all feature 
embeddings $\bm{r_{{u}}}$, where the weights 
$\bm{s_{{u}}}$ is the relevance between the features and the target.
Moreover, if we would like to highlight the most
related behavioral feature, $\bm{\hat{r}_{{u}}}$ can also be 
represented as
$$
\bm{\hat{r}_{{u}}}=
\texttt{softmax}(\bm{r_{{u}}} \bm{s_{{u}}})\in\mathbb{R^{{K}}}.
$$
Based on this, a CTR  prediction model with user interest selection can be 
formulated as
\begin{equation}\label{ctr_with_selector}
	\hat{{p}}({y}=1|{u},{v})=f_{
		\bm{\theta}}(\bm{\hat{r}_{{u}}},
	\bm{\theta}^{E}_{v}),
\end{equation}
where $\bm{s_{{u}}}$ is produced by $g_{\bm{\phi}}(\cdot)$ as 
Eq.(\ref{seletor}).

Notably, the relationship between 
$g_{\bm{\phi}}(\cdot)$ and $f_{\bm{\theta}}(\cdot)$
in this paper is a little different from other works. In most other work,  
$f_{\bm{\theta}}(\cdot)$ represents the whole model and 
$g_{\bm{\phi}}(\cdot)$ is merely its component. On the contrary, we regard 
$f_{\bm{\theta}}(\cdot)$ as a \textbf{base predictor} that can be applied with
arbitrary \textbf{feature selector} $g_{\bm{\phi}}(\cdot)$. They are viewed as
two independent modules that can be trained jointly or separately.

\section{Motivation}

\begin{figure*}[htbp]
	\centering
	\includegraphics[width=\linewidth]{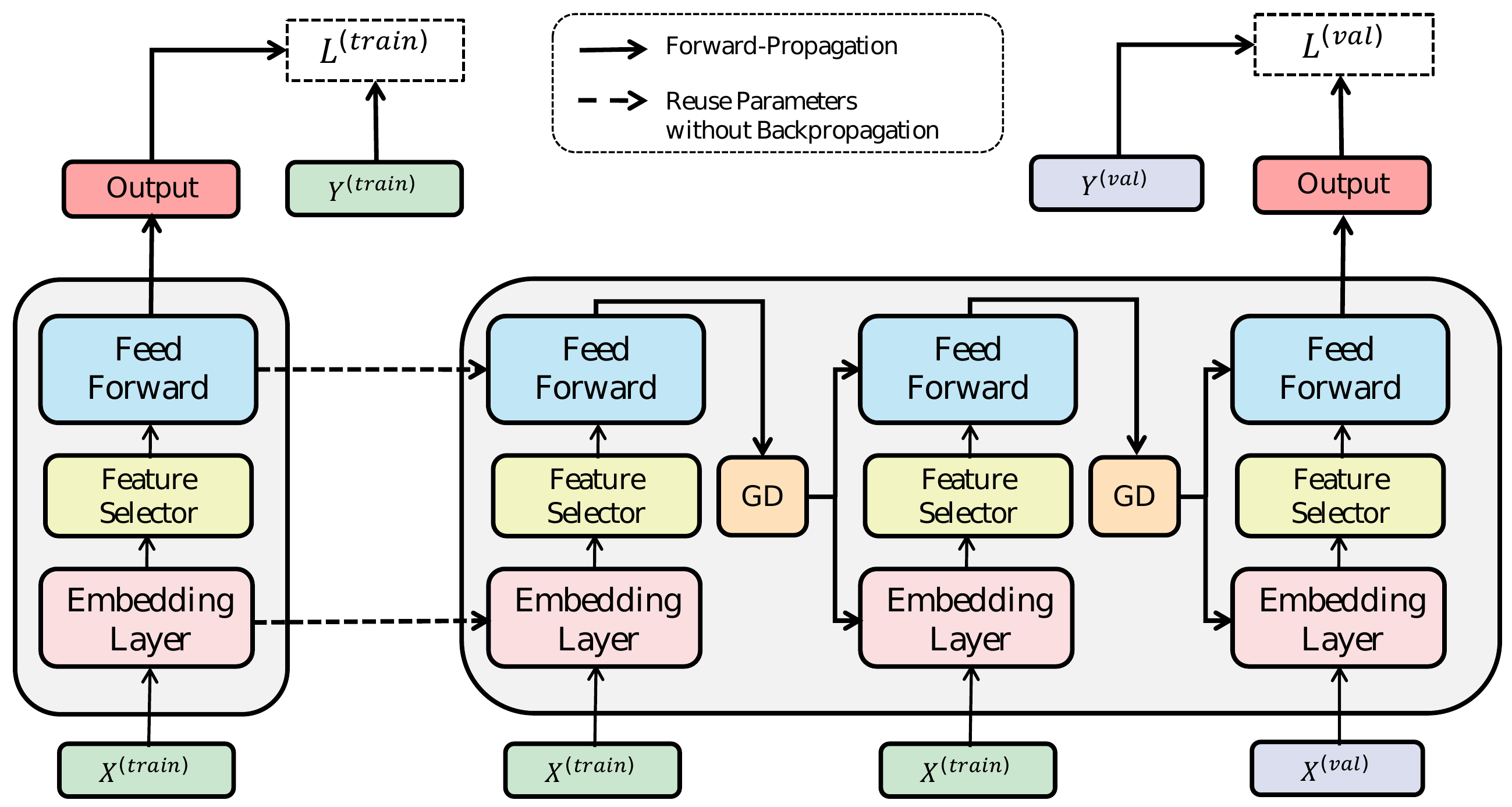}
	\caption{Overview of Meta-Wrapper Method (${N}=2$ for Example).} 
	\label{fig:Method} 
\end{figure*}

Normally, the feature selector $g_{\bm{\phi}}(\cdot)$ is modeled as an 
attention network that locally activates some behavioral features according to
the target item \cite{DIN}. Such an attention mechanism actually enhances the 
ability to fit the training data. Unfortunately, it also increases the 
overfitting risk, due to the over-parameterization.
Motivated by this fact, we first go back to a classical feature selection approach,
i.e., wrapper method \cite{fswrapper}. In this method, the feature selector is not 
learned from the training set, but from out-of-bag data. By incorporating the 
knowledge from different sources, it leads to a better generalization 
performance. 

Corresponding to user interest selection, 
the goal of the wrapper method can be formally represented as the following
bilevel optimization \cite{bilevel_meta}:
\begin{equation}\label{traditional_wrapper}
	\begin{aligned}
		  & \mathop{\arg\min}_{ \bm{\phi} }              
		\ L^{(out)}(
		\bar{\bm{\theta}}, 
		\bm{\phi}),\\
		  & \ \ s.t. \ \bar{\bm{\theta}}= 
		\mathop{\arg\min}_{ \bm{\theta} }\ 
		L^{(in)}(\bm{\theta}, \bm{\phi}), \\
	\end{aligned}
\end{equation}
where $L^{(\cdot)}(\cdot, \cdot)$ is the 
cross-entropy loss Eq.(\ref{cross_entropy}) with label ${y}$ and 
prediction $\hat{{p}}$ yield by Eq.(\ref{ctr_with_selector});
superscripts $(in)$ and $(out)$ indicate the training (in-the-bag) 
and out-of-bag data $\mathcal{D}^{(in)}$ and $\mathcal{D}^{(out)}$
that are also called inner- and outer-level dataset in this paper; 
$\bar{\bm{\theta}}$ represents
the parameters of an optimal $f_{\bm{\theta}}(\cdot)$ w.r.t a predefined 
$g_{\bm{\phi}}(\cdot)$. For such a method, the feature selection is motivated
by the need for generalization well to out-of-bag data  $\mathcal{D}^{(out)}$.

\begin{remark}[Feature Selector in Wrapper Method]
	Notably, canonical wrapper methods restrict
	the feature selector as a learnable indicating vector:
	$$
	g_{\bm{\phi}}(
	\{\bm{\theta}^{E}_{v_i}\}_{{i}=1}^{T_u}
	, \bm{\theta}^{E}_{v})
	=\bm{\phi}_{{u},{v}}\in\{0,1\}^{{T_u}}
	$$
	which is defined as
	$$
	\bm{\phi}_{{u}, {v}}[{i}]=\left\{
	\begin{array}{lr}
		1,\ \ {u}\ has\ clicked\ {v_i}; \\
		0,\ \ otherwise.                              \\
	\end{array}
	\right.
	$$
	Obviously, such a simple model is incapable of producing different 
	relevant scores 
	depending on different target items. Thus it is not suitable for user interest
	selection tasks. For a fair comparison with other methods, we will relax this 
	restriction in this paper, allowing the use of neural networks as 
	feature selectors.
\end{remark}

Unfortunately, the wrapper method actually cannot be directly applied 
to our problem, 
since its computational complexity is unacceptable. To explain, we provide
the following example.  

\begin{Example}[GDmax-Wrapper]\label{example_wrapper}
	Given a user ${u}$ and target item ${v}$, we are required to
	find an optimal feature selector $g_{\bm{\phi}}(\cdot)$ by
	solving the problem Eq.(\ref{traditional_wrapper}). for the sake of
	sufficient flexibility, $g_{\bm{\phi}}(\cdot)$ is modeled as a 
	neural network. Now we reverse
	the inner minimization as its equivalent maximization. 
	Thus Eq.(\ref{traditional_wrapper})
	is reformulated as the following:
	$$
	\begin{aligned}
		  & \mathop{\arg\min}_{ \bm{\phi} }              
		\ L^{(out)}(
		\bar{\bm{\theta}}, 
		\bm{\phi}),\\
		  & \ \ s.t. \ \bar{\bm{\theta}}= 
		\mathop{\arg\max}_{ \bm{\theta} }\ 
		\left(-L^{(in)}(\bm{\theta}, \bm{\phi})\right). \\
	\end{aligned}
	$$
	Following the paradigm of traditional wrapper method, we first get a 
	sufficiently
	precise inner solution $\bar{\bm{\theta}}$
	by performing gradient descent on $\bm{\theta}$ until $L^{(in)}$ 
	converges.
	Based on the inner solution, we solve the outer problem by optimizing
	$\bm{\phi}$. We find this process is quite similar to GDmax
	algorithm \cite{GDMAX}. Borrowing this concept, we call such a 
	wrapper method as GDmax-Wrapper.
\end{Example}
In this case, whenever we calculate $\bar{\bm{\theta}}$, 
we must scan the entire training set to solve the inner maximization. 
Assume that we need to perform $T^{(in)}$ loops of gradient ascent (GA) for inner
maximization and $T^{(out)}$ loops of gradient descent (GD) for outer
minimization.
Let $T^{(all)}$ notate the total number of gradient descent/ascent loops.
To find the optimal feature selector 
in Example.(\ref{example_wrapper}), we have $\mathbb{E}(T^{(all)})=
\mathcal{O}(T^{(in)} \times T^{(out)})$, where the computational complexity 
is unbearable.



\section{Methodology}

In this section, we propose a differentiable wrapper method, i.e., Meta-Wrapper, 
for user interest selection in CTR prediction. Specifically, we first 
illustrate two points about the Meta-Wrapper:
\begin{itemize}
	\item How to recast traditional wrapper method Eq.(\ref{traditional_wrapper}) 
	      to a continuous bilevel optimization problem;
	\item How our $g_{\bm{\phi}}(\cdot)$ and $f_{\bm{\theta}}(\cdot)$ 
	      are jointly trained by using our proposed method.
\end{itemize}  
After that, we theoretically analyse how our proposed method prevents overfitting.
Finally, we detail our network architecture.

\subsection{Differentiable Wrapping Operator}
Under the framework of the wrapper method, we model the user 
interest selection as a bilevel optimization Eq.(\ref{traditional_wrapper}). 
Unfortunately, such a problem is hardly solved in an efficient way
since its objective function is non-differentiable. Recalling the 
analysis about our Example.(\ref{example_wrapper}), we know the reason lies in
the inner-level \texttt{argmin} function.
With this obstacle, whenever we compute the outer-level loss, 
$\bm{\theta}$ must be trained from scratch, 
leading to great computational costs. 

However, if we approximate the inner-level \texttt{argmin} as
a differentiable function, then the losses in both levels could be jointly
optimized. This will bring a great improvement in efficiency.
To this end, we regard the learning process of the inner-level problem 
in Eq.(\ref{traditional_wrapper}) as a learnable function from the 
meta-learning \cite{MAML} perspective. By expanding its optimizing rollouts, 
we can recast  Eq.(\ref{traditional_wrapper}) as a
continuous bilevel optimization.

More specifically, we first define a function ${U}_{\bm{\phi}}(\cdot)$ 
so that ${U}_{\bm{\phi}}(\bm{\theta}, \mathcal{D}^{(in)})$ is a 
result of ${N}$-step gradient descent (GD) on dataset 
$\mathcal{D}^{(in)}$ for minimizing the inner loss:
\begin{equation}\label{learning-func}
	\begin{aligned}
		  & {U}_{\bm{\phi}}(\bm{\theta}, \mathcal{D}^{(in)}) 
		=\bm{\theta}^{({N})}, \\
	\end{aligned}
\end{equation}
where $\bm{\theta}^{({N})}$ is the optimal parameters at inner level.
To calculate $\bm{\theta}^{({N})}$, each GD step is performed as
\begin{equation}\label{GD}
	\begin{aligned}
		\bm{\theta}^{(j)}=                               
		\bm{\theta}^{(j-1)}                              
		-\beta \nabla_{\bm{\theta}^{(j-1)}}              
		L^{(in)}(\bm{\theta}^{(j-1)}, \bm{\phi}), 
	\end{aligned}
\end{equation}
where each $j\in\{1,\dots,{N}\}$ denotes the step number, thus 
$\bm{\theta}^{(0)}$ represents the initial value of $\bm{\theta}$;
$\beta > 0$ is the inner-level learning rate.

Based on this, the objective function of the wrapper method 
Eq.(\ref{traditional_wrapper}) can be 
recast to the following continuous optimization:
\begin{equation}\label{diff_wrapper}
	\begin{aligned}
		  & \mathop{\arg\min}_{ \bm{\phi} }              
		\ L^{(out)}(
		\bar{\bm{\theta}}(\bm{\phi}), 
		\bm{\phi}),\\
		  & \ \ s.t. \ \bar{\bm{\theta}}(\bm{\phi})= 
		{U}_{\bm{\phi}}(\bm{\theta}, \mathcal{D}^{(in)}), \\
	\end{aligned}
\end{equation}
where the inner-level \texttt{argmin} has been approximated as 
\begin{equation}\label{approx_argmin}
	\mathop{\arg\min}_{ \bm{\theta} }\ 
	L^{(in)}(\bm{\theta}, \bm{\phi}) \approx
	{U}_{\bm{\phi}}(\bm{\theta}, \mathcal{D}^{(in)}). 
\end{equation}

Based on Eq.(\ref{learning-func}) and Eq.(\ref{GD}), the computing process of 
$\bar{\bm{\theta}}(\bm{\phi})=
{U}_{\bm{\phi}}(\bm{\theta}, \mathcal{D}^{(in)})$ 
can be formulated as Alg.\ref{calc_U}.

\begin{algorithm}[h]
	\caption{
		Inner-level Optimization as ${U}_{\bm{\phi}}(\cdot)$ 
	}
	\label{calc_U}
	\hspace*{\algorithmicindent} 
	\textbf{Input:} initial value $\bm{\theta}^{(0)}$\\
	\hspace*{\algorithmicindent} 
	\textbf{Input:} Training data $\mathcal{D}^{(in)}$\\
	\hspace*{\algorithmicindent} 
	\textbf{Output:} ${N}$-step GD result $\bm{\theta}^{({N})}$ 
	depending on $\bm{\phi}$
						
	\begin{algorithmic}[1]
		\STATE Detach Operation: $\bm{\theta}':=\bm{\theta}^{(0)}$;
		\FOR{$j \leftarrow 1$ to ${N}$}
		\STATE $\texttt{Grad}(\bm{\theta}'):=
		\nabla_{\bm{\theta}'}
		L^{(in)}(\bm{\theta}', \bm{\phi})$;  
		\STATE $\bm{\theta}^{(j)}:=
		\bm{\theta}' -\beta \cdot \texttt{Grad}(\bm{\theta}')$;
		\STATE $\bm{\theta}':=\bm{\theta}^{(j)}$; 
		\ENDFOR
		\RETURN $\bm{\theta}^{({N})}$;
		\\
	\end{algorithmic}
\end{algorithm}

Notably, Alg.\ref{calc_U} is a functional process where the value of 
$\bm{\theta}$ would not be changed. Each GD-step is actually 
performed on a newly created node in the computational graph.
Hence we can apply the chain rule over its 
nested GD-loops to calculate $\nabla_{\bm{\phi}}{U}_{\bm{\phi}}(
\bm{\theta}, \mathcal{D}^{(in)})$, 
which is shown in Alg.\ref{back_U}.

\begin{algorithm}[h]
	\caption{
		backpropagation of ${U}_{\bm{\phi}}(\cdot)$ 
	}
	\label{back_U}
	\hspace*{\algorithmicindent} 
	\textbf{Input:} $\{\bm{\theta}^{({1})},\dots,
	\bm{\theta}^{({N})}\}$\\
	\hspace*{\algorithmicindent} 
	\textbf{Output:} $\nabla_{\bm{\phi}}\left(
	{U}_{\bm{\phi}}(\bm{\theta}, \mathcal{D}^{(in)})
	\right)$
						
	\begin{algorithmic}[1]
		\STATE Set $b:=
		\nabla_{\bm{\phi}}
		{U}_{\bm{\phi}}
		(\bm{\theta}^{({N})}, \bm{\phi})$
		\FOR{$j \leftarrow {N}$ to $1$}
		\STATE Set $a:=\bm{\theta}^{({j})}$;
		\STATE Set $b:=b-
		\beta(\nabla^{2}_{a, \bm{\phi}}{U}_{\bm{\phi}}
		(a, \bm{\phi}))^\top b$;
		\ENDFOR
		\RETURN $b$;
		\\
	\end{algorithmic}
\end{algorithm}

In this way, we can regard ${U}_{\bm{\phi}}(\cdot)$ as 
a learnable meta-model \cite{HyperNetworks} with parameters $\bm{\phi}$. 
The meta-model takes a base predictor as input then generates another one that 
has been sufficiently trained. Under the framework of the wrapper method,
we call this meta-model as our \textbf{differentiable wrapping operator}. 
More intuitively, It is demonstrated as the biggest gray box 
in Fig.(\ref{fig:Method}). Without using additional parameters
other than $\bm{\phi}$, learning
${U}_{\bm{\phi}}(\cdot)$ by minimizing Eq.(\ref{diff_wrapper}) 
is equivalent to learning feature selector $g_{\bm{\phi}}(\cdot)$. 
Therefore, we call such a feature selection strategy as \textbf{Meta-Wrapper}
method and Eq.(\ref{diff_wrapper}) as Meta-Wrapper loss.

\subsection{Jointly Learning With Base Predictor}
Recall that our predictive model consists of two components, i.e., the base
predictor $f_{\bm{\theta}}(\cdot)$ and user interest selector 
$g_{\bm{\phi}}(\cdot)$.
To obtain the completed model, we still have to learn 
parameters $\bm{\theta}$ in addition to $\bm{\phi}$. 

To this end, 
we unify the original cross-entropy loss and 
Eq.(\ref{diff_wrapper}) as the final loss function:
\begin{equation}\label{a_plus_b}
	\begin{aligned}
		  & \mathop{\arg\min}_{ (\bm{\theta}, \bm{\phi}) } 
		\ \mathcal{L}(\bm{\theta}, \bm{\phi})=
		L^{(in)}(\bm{\theta}, \bm{\phi}) +  
		\mu L^{(out)}(
		\bar{\bm{\theta}}(\bm{\phi}), 
		\bm{\phi}),\\
		  & \ \ s.t. \ \bar{\bm{\theta}}(\bm{\phi})=       
		{U}_{\bm{\phi}}(\bm{\theta}, \mathcal{D}^{(in)}), \\
	\end{aligned}
\end{equation}
where $\mu\in[0,1]$ is a coefficient to balance the CTR prediction and user
interest selection. If we set $\mu=0$, the Meta-Wrapper loss then disappears and
it degrades into an ordinary CTR prediction method 
with an attentive feature selector.

\begin{figure}[htbp]
	\centering
	\includegraphics[width=\linewidth]{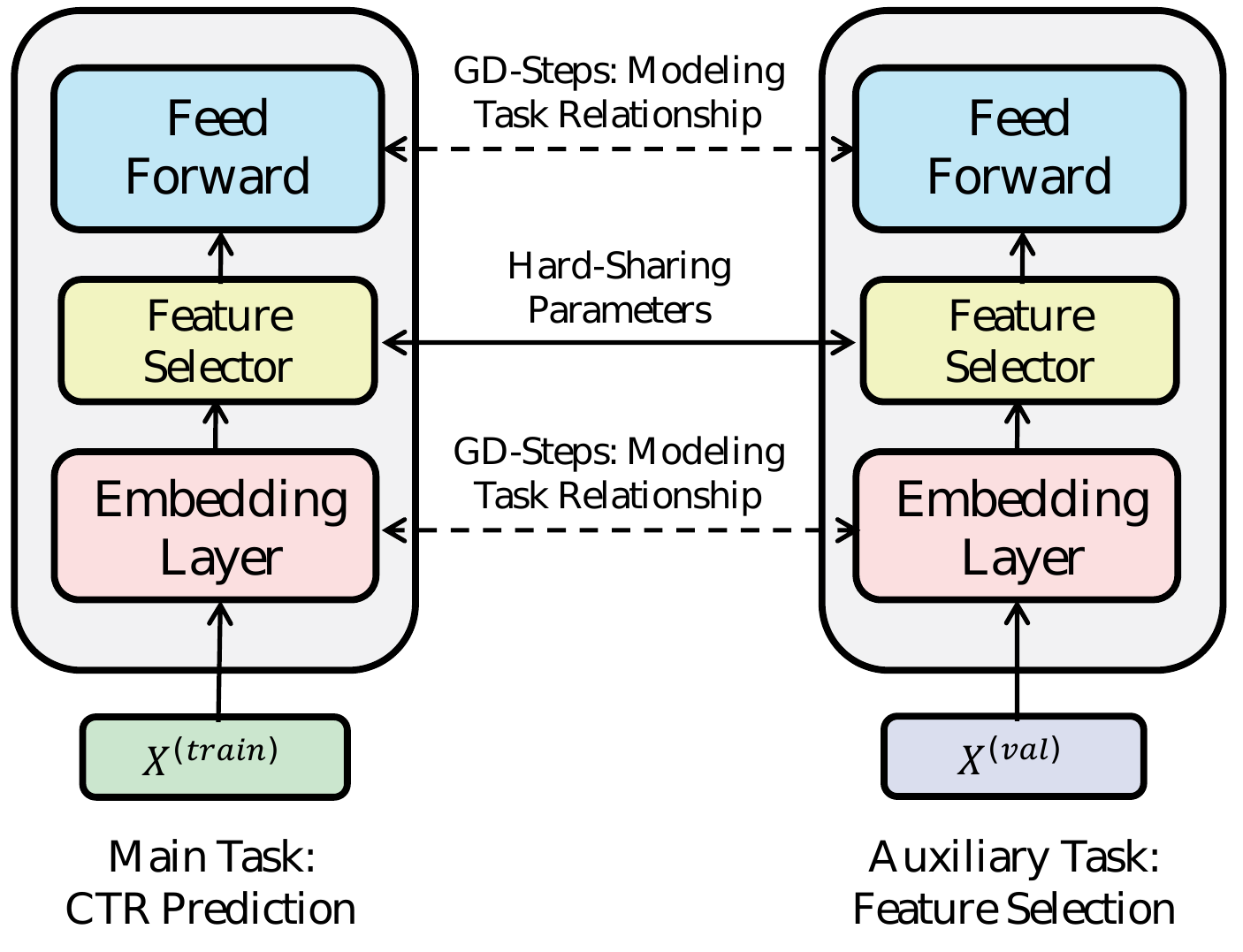}
	\caption{Multi-Task Learning Perspective.} 
	\label{fig:mtl} 
\end{figure}

From the multi-task learning perspective, we can regard CTR prediction as 
our main task and user interest selection as a corresponding auxiliary task.
Each task could transfer its learned knowledge to the other, thus all tasks can be
improved simultaneously. This is shown in Fig.(\ref{fig:mtl}). 
Compared with the traditional wrapper, our proposed method explicitly models the task
relationship by inner GD-loops, which is a more fine-grained manner than 
hard sharing the parameters.

For optimizing Eq.(\ref{a_plus_b}), the first step is to calculate 
$\bar{\bm{\theta}}(\bm{\phi})$ at inner level. After that, 
we plug the result into the outer-level loss and then perform minimization.
Thanks to the differentiability of ${U}_{\bm{\phi}}(\cdot)$, 
this problem can be efficiently solved by adopting mini-batch gradient descent 
\cite{SGD} as our optimizer. Specifically, we can perform inner-level 
optimization as Alg.\ref{calc_U}. By virtue of this, the optimal 
$\bm{\theta}$ and $\bm{\phi}$ can be obtained by minimizing the 
outer-level loss as Alg.\ref{training}, where the $\bm{\gamma}_i$ is the 
outer-level learning rate such that $\lim_{i \rightarrow \infty}\gamma_{i}=0$ and 
$\sum_{i=1}^{\infty}\gamma_{i}=\infty$. 
Meanwhile, Alg.\ref{back_U} performs
as an intermediate node in the computing graph when calculating the gradient
$
\nabla_{(\bm{\theta}, \bm{\phi})}
{\mathcal{L}(\bm{\theta}, \bm{\phi})}
$ in Alg.\ref{training}. Here the 
combination of $\{\mathcal{D}^{(in)}, \mathcal{D}^{(out)}\}$ can be viewed as
a meta-learning task $\mathcal{T}$. We represent 
the distribution of $\mathcal{T}$ as $\mathcal{P}(\mathcal{T})$ for the convenience of 
further analysis.

\begin{algorithm}[h]
	\caption{Solving Eq.(\ref{a_plus_b}) by  Mini-batch GD }
	\label{training}
	\hspace*{\algorithmicindent} 
	\textbf{Input:} Entire Inner-level Dataset $\bm{\Omega}^{(in)}$\\
	\hspace*{\algorithmicindent} 
	\textbf{Input:} Entire Outer-level Dataset $\bm{\Omega}^{(out)}$\\
	\hspace*{\algorithmicindent} 
	\textbf{Input:} Initial parameters $\bm{\phi}_{(0)}, \bm{\theta}_{(0)}$\\
	\hspace*{\algorithmicindent} 
	\textbf{Input:} Number of epochs ${T} \in \mathbbm{N}$\\
	\hspace*{\algorithmicindent} 
	\textbf{Output:} 
	Learned parameters 
	$\bm{\phi}_{({T})}, \bm{\theta}_{({T})}$
						
	\begin{algorithmic}[1]
											
		\STATE Set  $\bm{\phi}:=\bm{\phi}_{(0)}$;
		\STATE Set $\bm{\theta}:=\bm{\theta}_{(0)}$;
		\FOR{$i \leftarrow 1,2,\ldots, {T} $}
		\STATE $\mathcal{D}^{(in)} \leftarrow$  
		Draw a batch from $\bm{\Omega}^{(in)}$;
		\STATE $\mathcal{D}^{(out)} \leftarrow$  
		Draw a batch from $\bm{\Omega}^{(out)}$;
		\STATE 
		$\mathcal{L}(\bm{\theta}, \bm{\phi})
		:=
		L^{(in)}(\bm{\theta}, \bm{\phi}) +  
		\mu L^{(out)}(
		\bar{\bm{\theta}}(\bm{\phi}), 
		\bm{\phi})
		$;
		\STATE 
		$(\bm{\theta}_{(i)}, \bm{\phi}_{(i)})
		:=
		(\bm{\theta}, \bm{\phi}) - 
		\gamma_i \nabla_{(\bm{\theta}, \bm{\phi})}
		{\mathcal{L}(\bm{\theta}, \bm{\phi})}
		$;
		\STATE Set $\bm{\phi}:=\bm{\phi}_{(i)}$;
		\STATE Set $\bm{\theta}:=\bm{\theta}_{(i)}$;
		\ENDFOR
		\RETURN $\bm{\phi}_{({T})}, \bm{\theta}_{({T})}$
		\\
	\end{algorithmic}
\end{algorithm}

\subsection{Convergence Result}
In this section, we provide the convergence result of Alg.\ref{training} under a set
of broad assumptions. We find this algorithm could reach a stationary
point of the objective of our proposed method Eq.(\ref{a_plus_b}).
More formally, the algorithm would find a solution such that 
\begin{equation}\label{converge}
	\lim_{i \rightarrow \infty} 
	\|
	\nabla_{\bm{z}}\mathcal{F}_{i}
	\|^2
	= 0,
\end{equation}
where 
$\Vert \cdot \Vert$ represents the $L_2$-norm;
$\bm{z}=[\bm{\theta}^\top,\bm{\phi}^\top]^\top \in \mathbb{R}^p$ represents 
all of the parameters; $
\bm{z}_{(i)}$ is the value of the parameters after $i$th-step of 
\textbf{outer-level GD}, namely, \textbf{line:$7$ in Alg.\ref{training}};
$\mathcal{F}_{i}$ represents $\mathbb{E}_{p(\mathcal{T})}({\mathcal{L}(\bm{z}_{(i)})})$.
This equation means that there exists
an iteration of Alg.\ref{training}, which approaches a stationary
point to any level of proximity. 
More specifically, we first provide the assumption used for our proofs.
\begin{assumption}\label{lipchiz}
	For any mini-batch data $\mathcal{D}^{(in)} $ or $\mathcal{D}^{(out)} $, 
	$L^{(\cdot)}(\bm{z})$ are second-order differentiable
	as the function of $\bm{z}$, where $L$ is the cross-entropy loss on top of
	a neural network. Given arbitrary 
	$\bm{z},\bm{z}_{w},\bm{z}_{u} \in \mathbb{R}^p$,
	there are constants $C_0, C_1, C_2> 0$
	such that for any mini-batch data, we have
	\begin{enumerate}[label=(\arabic*)] 
		\item $\Vert \nabla_{\bm{z}}  L^{(\cdot)}(\bm{z}) \Vert \leq C_0$;
		\item $\Vert \nabla^2_{\bm{z}} L^{(\cdot)}(\bm{z}) \Vert \leq C_1$;
		\item $\Vert 
				\nabla^2_{\bm{\bm{z}}}L^{(\cdot)}(\bm{z}_{w})-
				\nabla^2_{\bm{\bm{z}}}L^{(\cdot)}(\bm{z}_{u})
			   \Vert \leq C_2 \Vert \bm{z}_{w} - \bm{z}_{u} \Vert$.
	\end{enumerate}
\end{assumption}
With this assumption, we can provide a convergence analysis in
the sense that a stationary point can be achieved within a finite number of 
iterations via Alg.\ref{training}. 
The readers are referred to Appendix for the proof of theorem.
\begin{theorem}[Convergence Result of Alg.\ref{training}]\label{converge_thm}
	Let $\mathcal{L}$ be the loss function defined as Eq.(\ref{a_plus_b}), 
	$L^{(in)}$ and $L^{(out)}$ be cross-entropy loss satisfying 
	Assumption.(\ref{lipchiz}),
	$C_r>0$ be a positive constant.
	Denote parameters
	$\bm{z} \in \mathbb{R}^p$ such that $\Vert \bm{z} \Vert^2 < C_r$,
	$
	C_3=N \beta C_2 (\sqrt{p}+\beta C_1)^{{N}-1}C_0
	$,
	$
	C_{L1}=2 C_1 + \mu C_3
	$,
	$
	C_{L2}=C^2_0 + \mu (\sqrt{p}+\beta C_1)^{{2N}}C_{0}^2
	$,
	$
	C_{L3}=\frac{C_{L1}}{2} (C^2_0 + C_{L2})
	$.
	Moreover, 
	let $\{\bm{z}_{(i)}\}_{i=1}^{\infty}$ be the parameters sequence generated by 
	applying
	Alg.\ref{training} for optimizing $\mathcal{L}(\bm{z})$,
	$\mathcal{F}_{i}=\mathbb{E}_{p(\mathcal{T})}({\mathcal{L}(\bm{z}_{(i)})})$,
	$\{\gamma_{i} > 0 \}_{i=1}^{\infty}$ be the sequence of outer-level learning rate. 
	Then it holds that
	\begin{enumerate}[label=(\arabic*)] 
		\item If $\forall i \in \mathbb{N}: \gamma_i < 2 \frac{C^2_0+C_{L2}}{C_{L3}}$
		, then $\{\mathcal{F}_i\}_{i=1}^{\infty}$ is a
		non-increasing sequence.
		\item If
		$\sum_{i=1}^{\infty}\gamma_{i}^2<\infty$, then
		$\lim_{i \rightarrow \infty} 
		\|
		\nabla_{\bm{z}}\mathcal{F}_{i}
		\|^2
		= 0$ and 
		there exists a subsequence of $\{\bm{z}_{s_i}\}_{i=1}^{\infty}$
		that converges to a stationary point.
		\item If
		$\forall k \in \mathbb{N}: \gamma_k = k^{-0.5}$, 
		then for any $\epsilon > 0$ we can get that
		$\min_{0 \leq j < k} \|\nabla_{\bm{z}} 
		{\mathcal{F}_{j}}\|^2
		=\bm{\Omega}(k^{-0.5+\epsilon})$.
	\end{enumerate}
\end{theorem}
					




\subsection{Implicit Regularization}\label{sec:theory}
In this section, we explore how our proposed Meta-Wrapper loss 
Eq.(\ref{diff_wrapper}) alleviates overfitting 
while learning the user interest selector.

Recall that our model consists of two components 1) base predictor,
and 2) user interest selector, i.e., wrapping operator, that is used to select user 
interested features (items).
We use $\bm{\theta}$ and $\bm{\phi}$ to represent vectorized parameters of 
these two components, respectively. For the convenience of further analysis, 
we define 
$$G_{\bm{\theta}}(\bm{\phi})=\nabla_{\bm{\theta}}{
	L^{(in)}(\bm{\theta}, \bm{\phi}})$$
and 
$$G'_{\bm{\theta}}(\bm{\phi})=\nabla_{\bm{\theta}}{
	L^{(out)}(\bm{\theta}, \bm{\phi}})$$
in this section.

Now let us analyse the Meta-Wrapper loss. With the inner gradient descent 
at the $k$-th step 
\begin{equation}\label{inner_gd}
	\bm{\bm{\theta}}^{(k)}=
	\bm{\bm{\theta}}^{(k-1)}-\beta G_{\bm{\bm{\theta}}^{(k-1)}}(\bm{\phi}),
\end{equation}
the outer loss can be reformulated as
\begin{equation}\label{outer_loss_inner_gd}
	\begin{aligned}
		  & L^{(out)}( 
		\bar{\bm{\theta}}^{(N)}, \bm{\phi})\\
		= &            
		L^{(out)}(\bm{\theta}^{{(N-1)}}-
		\beta G_{\bm{\bm{\theta}}^{{(N-1)}}}(\bm{\phi}), \bm{\phi}) 
	\end{aligned}
\end{equation}
where ${N}$ is the total number of inner updating loops. 

Then we can consider what 
Eq.(\ref{outer_loss_inner_gd}) could do for better resistance towards 
overfitting. We take the first-order Taylor expansion of 
$L^{(out)}(\bar{\bm{\theta}}, \bm{\phi})$ 
with respect to parameter $\bm{\theta}$ on 
$\bm{\theta}=\bm{\theta}^{(N)}$, 
yielding the following equation:
\begin{equation}\label{outer_loss_taylor_1}
	\begin{aligned}
		        & L^{(out)}(\bm{\theta}^{{(N-1)}}- 
		\beta G_{
		\bm{\bm{\theta}}^{{(N-1)}}}
		(\bm{\phi}), \bm{\phi})\\
		\approx &                                             
		L^{(out)}(
		{\bm{\theta}}^{(N-1)}, \bm{\phi})-
		\beta G_{\bm{\theta}^{({N-1})}}^\top (\bm{\phi}) 
		G'_{\bm{\theta}^{({N-1})}}(\bm{\phi}).
	\end{aligned}
\end{equation}
By further expanding Eq.(\ref{outer_loss_taylor_1}) recursively throughout the backpropagation, we ultimately reach to: 
\begin{equation}\label{outer_loss_taylor_n}
	\begin{aligned}
		        & L^{(out)}( 
		\bar{\bm{\theta}}^{(N-1)}, \bm{\phi})-
		\beta G_{\bm{\theta}^{({N-1})}}^\top (\bm{\phi}) 
		G'_{\bm{\theta}^{({N-1})}}(\bm{\phi}) \\
		\approx &            
		L^{(out)}(
		{\bm{\theta}}^{(N-2)}, \bm{\phi})-
		\beta \sum_{j=1}^{2}
		G_{\bm{\theta}^{({N-j})}}^\top (\bm{\phi}) 
		G'_{\bm{\theta}^{({N-j})}}(\bm{\phi}) \\
		\approx &            
		L^{(out)}(
		{\bm{\theta}}^{(N-3)}, \bm{\phi})-
		\beta \sum_{j=1}^{3}
		G_{\bm{\theta}^{({N-j})}}^\top (\bm{\phi}) 
		G'_{\bm{\theta}^{({N-j})}}(\bm{\phi}) \\
		        & \dots      \\
		\approx &            
		L^{(out)}(
		{\bm{\theta}}^{(0)}, \bm{\phi})
		\underbrace{
		- \beta \sum_{j=1}^{{N}}
		G_{\bm{\theta}^{({N-j})}}^\top (\bm{\phi}) 
		G'_{\bm{\theta}^{({N-j})}}(\bm{\phi})}_{\Delta},
	\end{aligned}
\end{equation}
where $\bm{\theta}^{(0)}$ is the inner initial value of 
$\bm{\theta}$ before updating 
the parameters on the current batch. Eq.(\ref{outer_loss_taylor_n}) reveals that the Meta-Wrapper loss could be approximately decomposed as a sum of two 
principled terms. Here the first term simply 
measures the out-of-bag performance of the attentive user interest selector. 
Notably, the second term of Eq.(\ref{outer_loss_taylor_n}) is the sum 
of negative inner products between $\nabla_{\bm{\theta}}{L^{(in)}}$ and 
$\nabla_{\bm{\theta}}{L^{(out)}}$ w.r.t different $\bm{\phi}$. More precisely,  
$L^{(in)}$ and $L^{(out)}$ are cross-entropy losses measured respectively  
on inner- and outer-level datasets. 

In a reasonable case, if both datasets are free of noises and subject 
to the same  distribution, the inner
product should be large since the intrinsic patterns of user interest 
should be dataset-invariant.
However, when noises are presented, the gradients across datasets become 
less similar, leading to a small inner product. In this way, $\Delta$ 
appeared in the last line of Eq.(\ref{outer_loss_taylor_n}) could be 
regarded as an implicit regularization to penalize gradient inconsistency due to the noises. 

Note that once the inner loop is introduced, we have $N \geq 1$, then
$\Delta$ in Eq.(\ref{outer_loss_taylor_n}) becomes nonzero and thus 
starts to alleviate the overfitting. Compared with traditional wrapper method shown in
Example.(\ref{example_wrapper}), we have $N \ll T^{(in)}$. In other words,
we can effectively control overfitting in our meta-wrapper with much smaller inner-GD loops than what 
are required traditionally.

\subsection{Computational Efficiency}

Compared with traditional attention mechanism, the extra computational 
burden of our proposed method lies in the backpropagation of $L^{(out)}$ in Eq.(\ref{a_plus_b}).
Thus we first derive its gradients:

\begin{equation}\label{chain_L_out}
	\begin{aligned}
		& \nabla_{\bm{\phi}}L^{(out)}( 
		  {\bm{\theta}}^{(N)}, \bm{\phi})\\
		= &        
		(\nabla_{\bm{\phi}}{\bm{\theta}}^{(N)})^\top
		\bm{\hat{g}}_{1}(N) + \bm{\hat{g}}_{2}(N)\\
		= &
		\nabla^\top_{\bm{\phi}}[\bm{\theta}^{(N-1)} - \beta \bm{g}_{1}(N-1)]
		\bm{\hat{g}}_{1}(N) + \bm{\hat{g}}_{2}(N) \\
		= &
		\nabla^\top_{\bm{\phi}}[\bm{\theta}^{(0)} - \beta \sum_{k=0}^{N-1} \bm{g}_{1}(k)]
		\bm{\hat{g}}_{1}(N) + \bm{\hat{g}}_{2}(N) \\
		= &
		- \beta \sum_{k=0}^{N-1} \underbrace{\left[ \bm{H}^\top_{1,2}(k)\bm{\hat{g}}_{1}(N) \right]}_{\Lambda_k}
		+ \bm{\hat{g}}_{2}(N) \\
	\end{aligned}
\end{equation}
where
$$\bm{\hat{g}}_{1}(N)=
\frac{\partial L^{(out)}({\bm{\theta}}^{(N)}, \bm{\phi})}{\partial {\bm{\theta}}^{(N)}}
\in \mathbb{R}^{N_\theta}$$
and 
$$\bm{\hat{g}}_{2}(N)=
\frac{\partial L^{(out)}({\bm{\theta}}^{(N)}, \bm{\phi})}{\partial \bm{\phi}}
\in \mathbb{R}^{N_\phi}$$
represent partial gradients at outer level;
$$\bm{g}_{1}(k)=
\frac{\partial L^{(in)}({\bm{\theta}}^{(k)}, \bm{\phi})}{\partial {\bm{\theta}}^{(k)}}
\in \mathbb{R}^{N_\theta}$$
is the partial gradient of ${\bm{\theta}}^{(k)}$ 
for any $k \in \{0, 1, \dots, N-1\}$
at inner level;
$$
\bm{H}_{1,2}(k)=\nabla_{\bm{\phi}}
\left(\frac{\partial L^{(in)}({\bm{\theta}}^{(k)}, \bm{\phi})}{\partial {\bm{\theta}}^{(k)}}
\right)
\in \mathbb{R}^{N_{\theta} \times N_{\phi}}
$$
for any $k \in \{0, 1, \dots, N-1\}$ is a block in Hessian matrix. Moreover, 
we define 
$$\Lambda_k=\bm{H}^\top_{1,2}(k)\bm{\hat{g}}_{1}(N)
\in \mathbb{R}^{N_{\phi}}$$ for the convenience
of further analysis.

Now we discuss the efficiency of computing 
$
\Lambda_k
$ by using the reverse mode of automatic differentiation. 
The automatic differentiation is
widely used in modern deep learning tools like 
TensorFlow \cite{tf} and PyTorch \cite{pytorch}.
Recall that for a differentiable scalar function $f(\bm{x})$, the 
automatic differentiation mechanism could compute $\nabla f(\bm{x})$ 
in a comparable time of computing $f(\bm{x})$ itself.
Such a technique is called \textbf{Cheap Gradient Principle}  \cite{CGP}.
For simplicity, we omit its formal definitions and derivations which can 
be found in the dedicated literature \cite{CGP, reverse_acc}.
Through the lens of this principle, we can further conclude that
$\Lambda_k$ could be computed in an efficient manner.
Specifically, given an arbitrary $\bm{v} \in \mathbb{R}^{N_\theta}$,
we know the Hessian-vector product 
$\bm{H}^\top_{1,2}(k)\bm{v}$ is the gradient w.r.t $\bm{\phi}$ of 
the scalar function $\bm{g}_{1}(k)^\top\bm{v}$.
Thus $\bm{H}^\top_{1,2}(k)\bm{v}$ can be directly evaluated in a 
comparable time of $\bm{g}^\top_{1}(k)\bm{v}$ without 
explicitly constructing Hessian matrix. 
Substituting $\bm{v}$ for $\bm{\hat{g}}_{1}(N)$, 
we can conclude that $\Lambda_k$ could be computed in comparable time
with  $\bm{g}^\top_{1}(k)\bm{\hat{g}}_{1}(N)$.

Based on this, we can analyse the time complexity of Alg.\ref{training}.
Let $\mathcal{O}(C)$ denote the complexity for computing 
$L^{(in)}$ or $L^{(out)}$, then
the gradients and Hessian-vector product $\Lambda_k$ can be obtained 
within $\mathcal{O}(C)$ and $\mathcal{O}(N_\theta + C)$, respectively.
In this way, the time complexity of calculating $\nabla_{\bm{\phi}}L^{(out)}( 
{\bm{\theta}}^{(N)}(\bm{\phi}), \bm{\phi})$ is 
$\mathcal{O}(N \cdot N_\theta + N \cdot C + N_\phi)$, where
$N$ is the number of inner-GD loops.
Thus the overall gradients of $\mathcal{L}$ can be calculated
within 
$\mathcal{O}(N \cdot N_\theta + N \cdot C + N_\phi + C)$
which is no greater than
$\mathcal{O}(N \cdot N_p + N \cdot C + C)$.
Then the complexity of 
each outer-GD step is no greater than
$\mathcal{O}(N \cdot N_p + N \cdot C + N_p + C)$
where $N_p = N_\theta + N_\phi$.
Similarly, 
for traditional single-level optimization strategies, 
the time complexity of each outer GD 
would be $\mathcal{O}(N_p + C)$. Recall that $N$ can be set as a very 
small integer. Hence, the proposed method could achieve 
comparable computational efficiency with traditional attention mechanisms
during the learning process.

\subsection{Network Structure}\label{sec:net_struct}
As is mentioned in the previous sections, our model consists of two components:
1) the base predictor and 2) the feature selector, namely the wrapping operator.
In this section, we will detail the model architecture of both components
respectively. 

In this paper, we implement the base predictor as a simple neural
network, with an embedding layer and multiple
dense layers, where the parameters are denoted as $\bm{\theta}$ as is 
mentioned in previous sections. 
More specifically, the base predictor consists of 
the following components:
\begin{itemize}
	\item \textbf{Embedding Layer:} The sparse features, for each instance,
	are first transformed into low dimensional vectors by an embedding layer.
	For each feature $i$, the model uses a $K$-dimensional embedding vector 
	$\bm{\theta}^{E}_{v_{.}}\in\mathbb{R}^{{K}}$ as its representation.
	\item \textbf{User Modeling Layer:} Since the users have various numbers 
	of clicked items, the length of behavioral sequence is not fixed. 
	To handle this situation, we resort to a pooling layer which produces the
	fusion of user's behavioral features. This layer takes embeddings of 
	behavioral features as inputs, and then outputs the weighted sum of the
	embeddings as the user's representation. Here the weights are relevant scores
	between features and target item, which is generated by feature selector.
	\item \textbf{Dense Layers:} The user representation is concatenated with 
	other feature embeddings, going through multiple fully connected layers with
	\texttt{Sigmoid} activation. The last layer outputs the predicted CTR.
\end{itemize}

Similar to other works \cite{DIN, DIEN}, the feature selector is also implemented 
as a neural network, where the trainable parameters are represented as $\bm{\phi}$. 
The network consists of the following components:
\begin{itemize}
	\item \textbf{Interaction Layer:} 
	This layer takes the embeddings of a behavioral feature and 
	the target item as inputs. After that, it outputs the element-wise subtraction
	and production of these two embeddings. These operations could
	capture the second-order information of each feature and the target.
 	\item \textbf{Concatenation Layer:} 
	 In this layer, the model concatenates four vectors: 1) behavioral feature embedding,
	 2) target item embedding, 3) the subtraction output by the previous layer, and 
	 4) the production output by the previous layer.
	\item \textbf{Dense Layers:} 
	Going through multiple dense layers with \texttt{Sigmoid} activation, the model
	outputs the relevance score between a behavioral feature and the target item.
\end{itemize}



\section{Experiments}

\begin{table*}[htbp]
	\centering
	\caption{Performance Comparison On Public Datasets.}
	\setlength{\tabcolsep}{2.81mm}
	  \begin{tabular}{lrr|rr|rr|rr|rr}
			& \multicolumn{2}{c}{Electronics} & \multicolumn{2}{c}{Books} & \multicolumn{2}{c}{Games} & \multicolumn{2}{c}{Taobao} & \multicolumn{2}{c}{Movielens} \\
	  \midrule
			& AUC   & \multicolumn{1}{r}{Impr} & AUC   & \multicolumn{1}{r}{Impr} & AUC   & \multicolumn{1}{r}{Impr} & AUC   & \multicolumn{1}{r}{Impr} & AUC   & Impr \\
	  \midrule
	  Base  & \cellcolor[rgb]{ 1,  .976,  .953}0.7524 & 0.00\% & \cellcolor[rgb]{ 1,  .984,  .922}0.7616 & 0.00\% & \cellcolor[rgb]{ .922,  .984,  .961}0.7199 & 0.00\% & \cellcolor[rgb]{ .914,  .965,  .992}0.8298 & 0.00\% & \cellcolor[rgb]{ .914,  .941,  .969}0.8730 & 0.00\% \\
	  Wide\&Deep    & \cellcolor[rgb]{ 1,  .776,  .576}0.7937 & 16.36\% & \cellcolor[rgb]{ 1,  .937,  .706}0.7739 & 4.70\% & \cellcolor[rgb]{ .647,  .871,  .788}0.7404 & 9.32\% & \cellcolor[rgb]{ .788,  .902,  .965}0.8828 & 16.07\% & \cellcolor[rgb]{ .902,  .933,  .965}0.8739 & 0.24\% \\
	  PNN   & \cellcolor[rgb]{ 1,  .773,  .576}0.7940 & 16.48\% & \cellcolor[rgb]{ 1,  .949,  .765}0.7707 & 3.48\% & \cellcolor[rgb]{ .698,  .894,  .82}0.7384 & 8.43\% & \cellcolor[rgb]{ .71,  .863,  .949}0.8857 & 16.95\% & \cellcolor[rgb]{  .902,  .933,  .965}0.8739 & 0.24\% \\
	  DeepFM   & \cellcolor[rgb]{ 1,  .757,  .541}0.7979 & 18.03\% & \cellcolor[rgb]{ 1,  .961,  .816}0.7678 & 2.37\% & \cellcolor[rgb]{ .588,  .847,  .753}0.7425 & 10.28\% & \cellcolor[rgb]{ .804,  .91,  .969}0.8821 & 15.86\% & \cellcolor[rgb]{ .875,  .914,  .953}0.8754 & 0.64\% \\
	  xDeepFM  & \cellcolor[rgb]{ 1,  .843,  .702}0.7800 & 10.94\% & \cellcolor[rgb]{ 1,  .976,  .882}0.7640 & 0.92\% & \cellcolor[rgb]{ .922,  .984,  .961}0.7203 & 0.18\% & \cellcolor[rgb]{ .914,  .965,  .992}0.8459 & 4.88\% & \cellcolor[rgb]{ .78,  .851,  .918}0.8810 & 2.14\% \\
	  FGCNN & \cellcolor[rgb]{ 1,  .863,  .741}0.7757 & 9.23\% & \cellcolor[rgb]{ 1,  .984,  .922}0.7649 & 1.26\% & \cellcolor[rgb]{ .922,  .984,  .961}0.7229 & 1.36\% & \cellcolor[rgb]{ .914,  .965,  .992}0.8532 & 7.10\% & \cellcolor[rgb]{ .871,  .914,  .953}0.8757 & 0.72\% \\
	  DIN   & \cellcolor[rgb]{ 1,  .745,  .522}0.7999 & 18.82\% & \cellcolor[rgb]{ .996,  .906,  .569}0.7818 & 7.72\% & \cellcolor[rgb]{ .384,  .765,  .624}0.7502 & 13.78\% & \cellcolor[rgb]{ .447,  .729,  .89}0.8955 & 19.92\% & \cellcolor[rgb]{ .808,  .871,  .929}0.8794 & 1.72\% \\
	  DIEN  & \cellcolor[rgb]{ 1,  .753,  .537}0.7981 & 18.11\% & \cellcolor[rgb]{ 1,  .941,  .733}0.7724 & 4.13\% & \cellcolor[rgb]{ .388,  .765,  .624}0.7501 & 13.73\% & \cellcolor[rgb]{ .416,  .714,  .882}0.8967 & 20.29\% & \cellcolor[rgb]{ .584,  .722,  .839}0.8925 & 5.23\% \\
	  DMIN  & \cellcolor[rgb]{ 1,  .694,  .427}0.8171 & 25.63\% & \cellcolor[rgb]{ .996,  .914,  .604}0.7799 & 7.00\% & \cellcolor[rgb]{ .4,  .769,  .631}0.7497 & 13.55\% & \cellcolor[rgb]{ .78,  .898,  .965}0.8831 & 16.16\% & \cellcolor[rgb]{ .588,  .722,  .839}0.8924 & 5.20\% \\
	  AutoFIS    & \cellcolor[rgb]{ 1,  .773,  .573}0.7943 & 16.60\% & \cellcolor[rgb]{ .996,  .918,  .62}0.7789 & 6.61\% & \cellcolor[rgb]{ .431,  .784,  .651}0.7484 & 12.96\% & \cellcolor[rgb]{ .82,  .918,  .973}0.8816 & 15.71\% & \cellcolor[rgb]{ .902,  .933,  .965}0.8739 & 0.24\% \\
	  AutoGroup    & \cellcolor[rgb]{ 1,  .773,  .573}0.7944 & 16.64\% & \cellcolor[rgb]{ .996,  .918,  .62}0.7790 & 6.65\% & \cellcolor[rgb]{ .588,  .847,  .753}0.7425 & 10.28\% & \cellcolor[rgb]{ .804,  .91,  .969}0.8821 & 15.86\% & \cellcolor[rgb]{ .886,  .922,  .957}0.8748 & 0.48\% \\
	  Ours  & \cellcolor[rgb]{ 1,  .694,  .427}0.8201 & 26.82\% & \cellcolor[rgb]{ .992,  .871,  .424}0.8237 & 23.74\% & \cellcolor[rgb]{ .31,  .733,  .576}0.7530 & 15.05\% & \cellcolor[rgb]{ .298,  .655,  .855}0.9010 & 21.59\% & \cellcolor[rgb]{ .408,  .6,  .769}0.9030 & 8.04\% \\
	  \bottomrule
	  \end{tabular}%
	\label{tab:perform_tab}%
\end{table*}%

In this section, we aim at answering the following question:
\begin{description}
	\item[\textbf{RQ1}:] How is the performance of Meta-Wrapper compared with 
	existing CTR prediction models?
	\item[\textbf{RQ2}:] Do these improvements come from our proposed
	feature selection method? 
	\item[\textbf{RQ3}:] Does Meta-Wrapper have better generalization
	performance as our theoretical analysis?
	\item[\textbf{RQ4}:] 
		Does Meta-Wrapper have a comparable efficiency as other popular methods? 
	\item[\textbf{RQ5}:] How do the hyperparameters affect the performance of 
	Meta-Wrapper?
\end{description}
To this end, we first introduce the datasets and our experimental setting. 
After that, we elaborate on the evaluation metric and competitors. 
Finally, we demonstrate our experimental results on three public datasets, 
showing the effectiveness of our proposed method.

\begin{table}
	\centering
	\small
	\caption{Statistics of Datasets.}
	\begin{tabular*}{\hsize}{@{}@{\extracolsep{\fill}}lcccc@{}}
		\multirow{2}[2]{*}{Dataset}   & \multirow{2}[2]{*}{\#Users} & \multirow{2}[2]{*}{\#Items}   & \multirow{2}[2]{*}{\#Categories} & \multirow{2}[2]{*}{\#Instances} \\
		                              &                             &                               &                                  &                                 \\
		\toprule
		\multirow{2}[1]{*}{MovieLens} & \multirow{2}[1]{*}{6,040}   & \multirow{2}[1]{*}{3,900}     & \multirow{2}[1]{*}{20}           & \multirow{2}[1]{*}{1,000,209}   \\
		                              &                             &                               &                                  &                                 \\
		\multirow{2}[1]{*}{Electronic}    & \multirow{2}[1]{*}{192,403} & \multirow{2}[1]{*}{63,001}    & \multirow{2}[1]{*}{801}          & \multirow{2}[1]{*}{1,689,188}   \\
		                              &                             &                               &                                  &                                 \\
		\multirow{2}[1]{*}{Taobao}    & \multirow{2}[1]{*}{987,994} & \multirow{2}[1]{*}{4,162,024} & \multirow{2}[1]{*}{9,439}        & \multirow{2}[1]{*}{100,150,807} \\
		                              &                             &                               &                                  &                                 \\
		\multirow{2}[1]{*}{Games}    & \multirow{2}[1]{*}{24,303} & \multirow{2}[1]{*}{10,672} & \multirow{2}[1]{*}{68}        & \multirow{2}[1]{*}{231,780} \\
		                              &                             &                               &                                  &                                 \\
		\multirow{2}[1]{*}{Books}    & \multirow{2}[1]{*}{603,668} & \multirow{2}[1]{*}{367,982} & \multirow{2}[1]{*}{1,579}        & \multirow{2}[1]{*}{8,898,041} \\
		                              &                             &                               &                                  &                                 \\
		
		\bottomrule
	\end{tabular*}%
	\label{tab:data}%
\end{table}%

\subsection{Experimental Setting}

\subsubsection{Datasets}

In this paper, experiments are conducted on the following three public datasets 
with rich user behaviors. Moreover, the statistics of these datasets are shown in 
Tab.\ref{tab:data}.

\noindent
\textbf{MovieLens Dataset}.
It refers to MovieLens-1M dataset which is a
popular benchmark for CTR prediction \cite{DIN, DIEN}.
It contains 1,000,209 anonymous ratings of approximately 3,900 movies 
made by 6,040 MovieLens users who joined MovieLens. The item features
include genres, titles and descriptions. In our experiments, training data 
is generated with user interacted items and their categories for each user. 

\noindent
\textbf{Electronics Dataset}. 
This dataset is a collection of user logs over online products (items) 
from Amazon platform, which consists of product reviews and product metadata. 
In the field of recommender systems, this dataset is often selected as a benchmark 
in CTR prediction task \cite{DIN}. In this paper, we conduct experiments on a subset named 
Electronics. The subset contains about 200000 users, 60000 items, 800 categories, 
and 1700000 instances. Features include item ID, category ID, user-reviewed 
items list and category list. 

\noindent
\textbf{Books Dataset}. 
This subset is also collected from Amazon, which has the same data structure as Electronics. 
In some recent work \cite{DIEN, DSIN}, this dataset is also selected as a benchmark. 
Books subset contains about 600000 users, 400000 items, 1500 categories, 
and 9000000 instances. 

\noindent
\textbf{Games Dataset}. 
This is another subset of Amazon, which has the same data structure as Books and Electronics. 
Games subset contains about 20000 users, 10000 items, 70 categories, 
and 200000 instances. 

\noindent
\textbf{Taobao Dataset}. This is a dataset of user behaviors from the commercial 
platform of Taobao. It contains about 1 million randomly selected users with 
rich behaviors. Each instance in this 
dataset represents a specific user-item interaction, which involves user ID, 
item ID, item category ID, behavior type and timestamp. The behavior types 
include click, purchasing, adding an item to the shopping cart and favoring 
an item. In the experiments of this paper, we only take the click-behavior 
to model the user interests.

Electronics, Books, Games and Taobao are typical datasets collected from real-world applications. 
In these datasets, some users are inactive and produce relatively 
sparse behaviors during this long time range. As for MovieLens, 
though it has fewer instances, behavioral 
records for each user are even denser than the other two datasets.

\subsubsection{Competitors}

\begin{figure*}[htbp]
	\centering
	\subfloat[Electronics]{
		\includegraphics[scale=0.43, trim=0 0 6 0, clip]{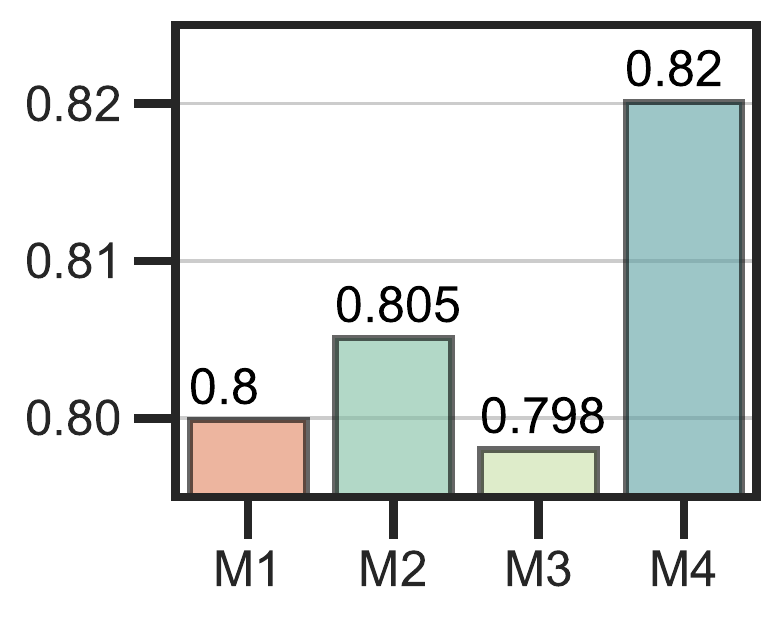}} 
	\subfloat[Books]{
		\includegraphics[scale=0.43, trim=0 0 6 0, clip]{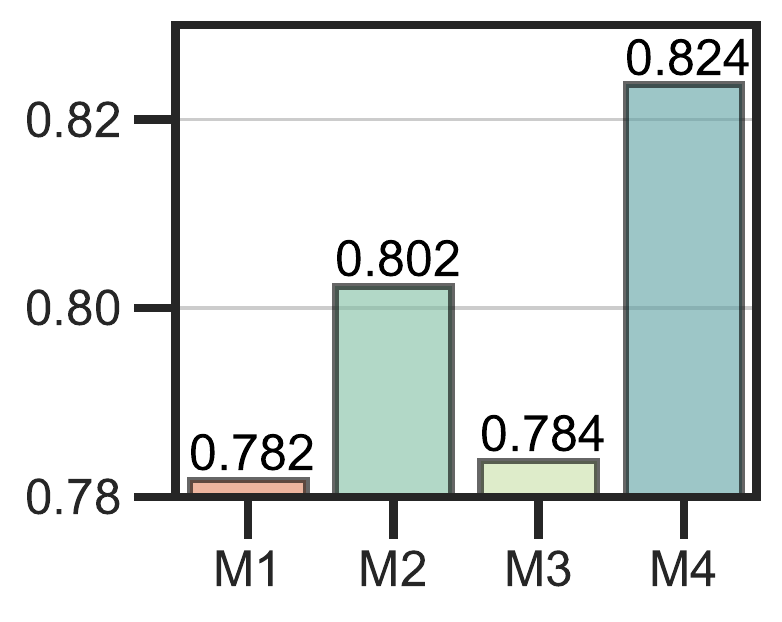}}
	\subfloat[Games]{
		\includegraphics[scale=0.43, trim=0 0 6 0, clip]{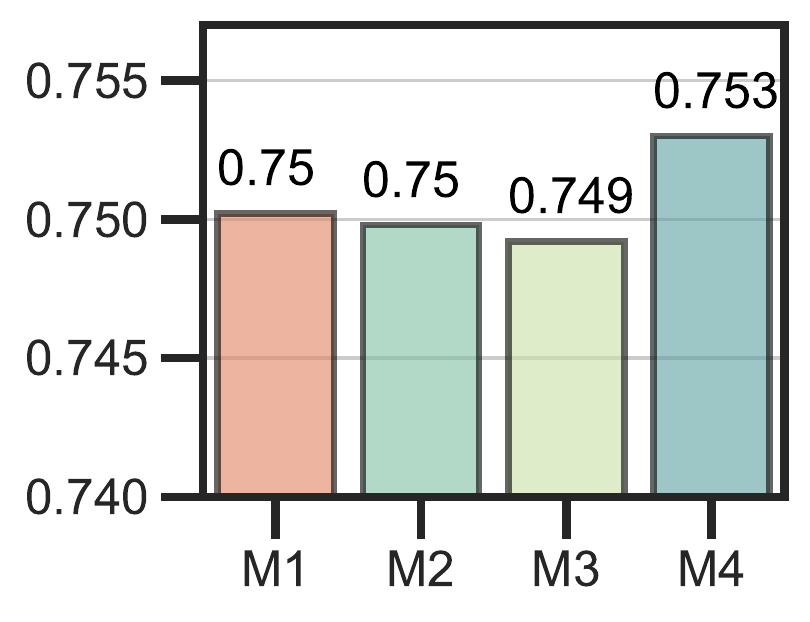}} 
	\subfloat[Taobao]{
		\includegraphics[scale=0.43, trim=0 0 6 0, clip]{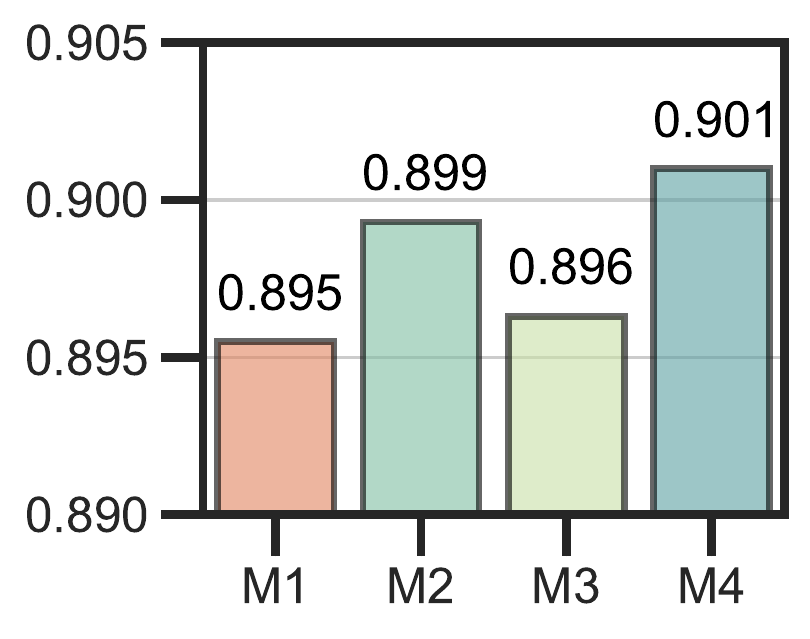}}
	\subfloat[MovieLens]{
		\includegraphics[scale=0.43, trim=0 0 6 0, clip]{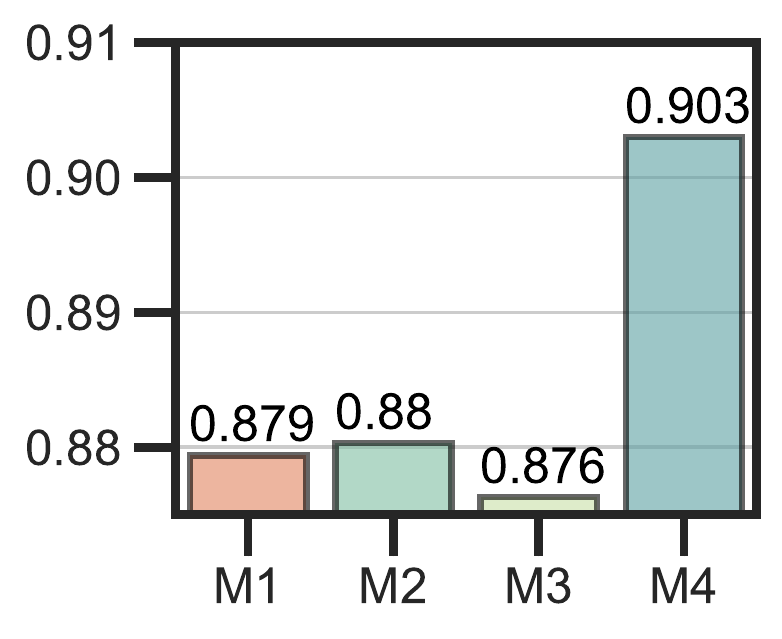}} \\
	\caption{Ablation Study On Public Datasets.} 
	\label{fig:albation}
\end{figure*}

In this paper, the following models are selected to make comparisons with our 
proposed method, i.e., Meta-Wrapper:
\begin{itemize}
	\item \textbf{Base}. This model uses the same Embedding+MLP 
		  architecture as the base predictor introduced in the preceding section,
		  where the user interest is simply modeled as the sum of clicked-item 
		  embeddings. This model is selected as a baseline in many 
		  studies such as \cite{DIN, DIEN, DMIN}.
	\item \textbf{Wide\&Deep}\cite{WD}. Wide\&Deep model is widely accepted in real-world
		  applications. 
	      It consists of two components: 1) wide module, which handles
		  the manually designed features by a logistic regression model, 
		  2) deep module,
	      which is equivalent
	      to the BaseModel.
	      Following the practice in \cite{DFM}, we take a cross-product 
	      of all features as wide inputs.
	\item \textbf{PNN}\cite{PNN}. PNN can be regarded as an improved version of 
	      BaseModel by adding a product layer after the embedding layer to model 
	      high-order feature interactions. It first inputs the feature embeddings into a 
	      dense layer and a product layer, then concatenates them together and uses 
	      other dense layers to get the output.
	\item \textbf{DeepFM}\cite{DFM}. It takes factorization machines (FM) as a wide 
	      module in Wide\&Deep to reduce the feature engineering efforts. It inputs 
	      feature embeddings into an FM and a deep component and then concatenates 
	      their outputs to the final output through a dense layer.
	\item \textbf{DIN}\cite{DIN}. DIN applies the attention mechanism on top of 
	      the BaseModel for selecting relevant features in user's behaviors. The attentive 
	      feature selector has the same architecture as ours introduced in the 
		  proceeding section.
	\item \textbf{DIEN}\cite{DIEN}. DIEN can be considered as an improved version 
		  of DIN. It uses GRU with an attentional updating gate to model the
		  evolution of user interest.
	\item \textbf{DMIN}\cite{DMIN}. DMIN uses multi-head self-attention to
		  model the interest evolution and extract 
		  multiple user interests within behavioral sequences. It is more efficient
		  at inference stage compared with DIEN.
	\item \textbf{xDeepFM}\cite{XDFM}. 
		  xDeepFM uses Compressed Interaction Network (CIN) 
		  to model the feature crossing. This approach could describe high-order feature
		  interactions in an explicit way. 
	\item \textbf{FGCNN}\cite{FGCNN}. FGCNN uses CNN to automatically generate new features,
		  thereby augmenting feature space. 
		  This can be regarded as novel method to model the high-order feature-crossing.
	\item \textbf{AutoFIS}\cite{AutoFIS}. AutoFIS is proposed to automatically select feature
		  interactions. In this paper, we use this strategy to select user behavioral features.
	\item \textbf{AutoGroup}\cite{AutoGroup} AutoGroup is also a feature selection based method.
	      It first aggregates features into groups, and then applies feature selection on top of
		  these groups. We also use this strategy to process user behavioral features in our 
		  experiments. 
\end{itemize}
\noindent
For all methods, two dense 
layers of the predictor are each with (80, 40) units. For DIN and our proposed
method, the two dense layers in feature selector have 80 and 40 units, respectively.
We adopt Mini-batch GD \cite{SGD} as the optimizer for all methods,
where the learning rate is dynamically scheduled by an exponential decay.
Meanwhile, the number of epochs is set to 30.
For fair comparison, the following hyperparameters are tuned by grid search:
learning rate $\gamma$ in $\{1, 0.1, 0.01, 0.0001, 1.2, 0.12, 0.012, 0.00012\}$; 
batch size in $\{32, 64, 128, 256, 512, 1024\}$;
weight decay rate in $\{0, 1e-8, 1e-7, 1e-6, 1e-5, 1e-4\}$;
embedding dimension in $\{32, 48, 64, 96, 128\}$.
For Meta-Wrapper in Eq.(\ref{a_plus_b}), we additionally tune 
$\mu$ in $\{0.2, 0.4, 0.6, 0.8\}$,
$\beta$ in $\{0.0001, 0.001, 0.01, 1\}$.
Moreover, we set ${N}=1$ when comparing with other methods, but we 
tune ${N}$ in $\{1, 2, 3\}$ during the hyperparameter study.


\subsubsection{Data Splitting}

To simulate the experimental setting, we first sort the clicked records for each 
user by the timestamp and
construct the training, validation and test set. 
Specifically, given ${T_u}$ behaviors 
for each user, we take 
last records as the positive samples in test set. 
From the rest records, we 
use the last one
for each user as positive samples in validation set, while the remaining ones 
belong to the training set. 
For those methods needing out-of-bag data, we randomly split the training set 
into two parts at each epoch. Specifically, we regard $80\%$ of the samples as
$\mathcal{D}^{(in)}$ in previous sections, while the rest $20\%$ of the samples 
are selected as $\mathcal{D}^{(out)}$.
In each dataset, the items of positive samples are randomly replaced with 
another non-clicked item to construct the negative samples. 
In our experiment, each model is trained on the training set and its
hyperparameters are tuned on the validation set. 
Notably, the validation set is not directly used to update parameters.

\begin{figure*}[htbp]
	\centering
	\subfloat[Training Time on Electronics]{
		\includegraphics[scale=0.43, trim=20 0 5 0, clip]{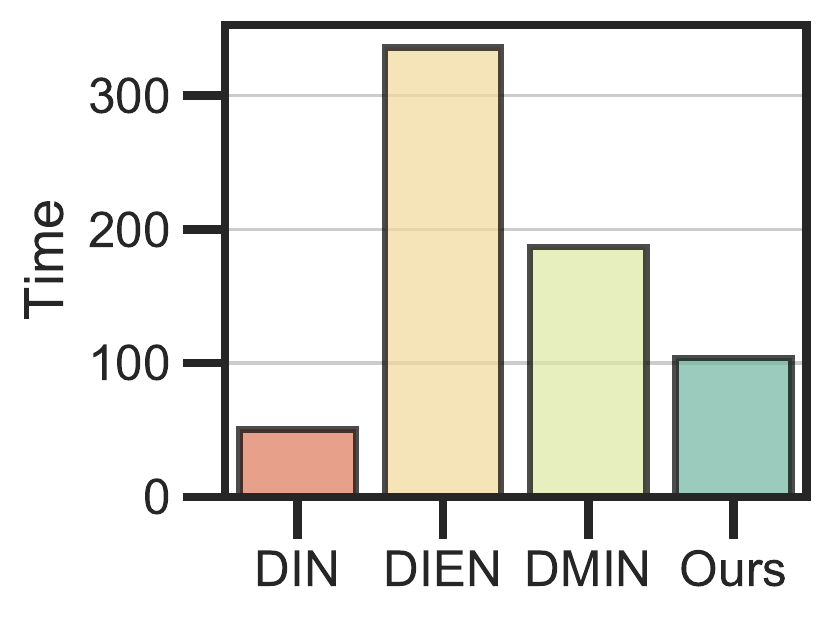}} 
	\subfloat[Training Time on \\Books]{
		\includegraphics[scale=0.43, trim=20 0 5 0, clip]{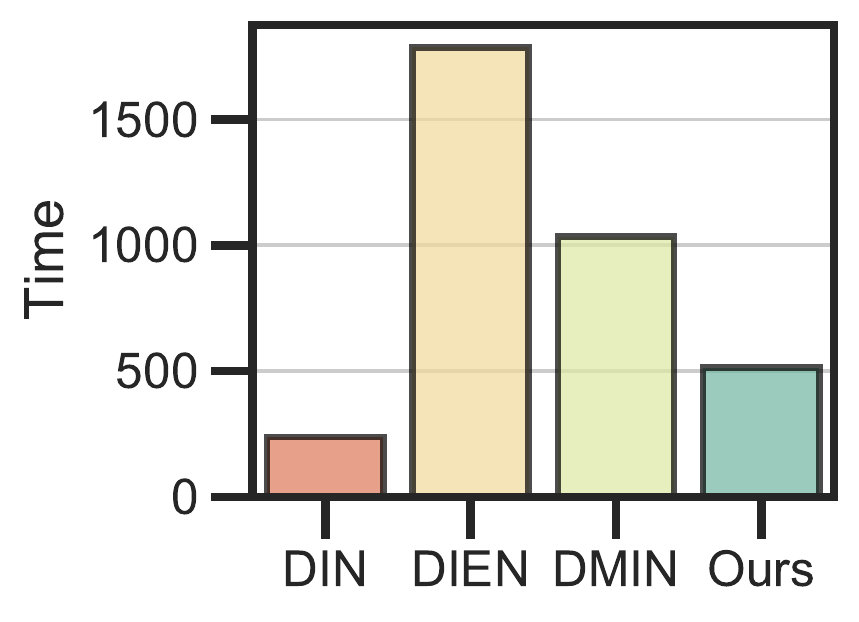}}
	\subfloat[Training Time on Games]{
		\includegraphics[scale=0.43, trim=20 0 5 0, clip]{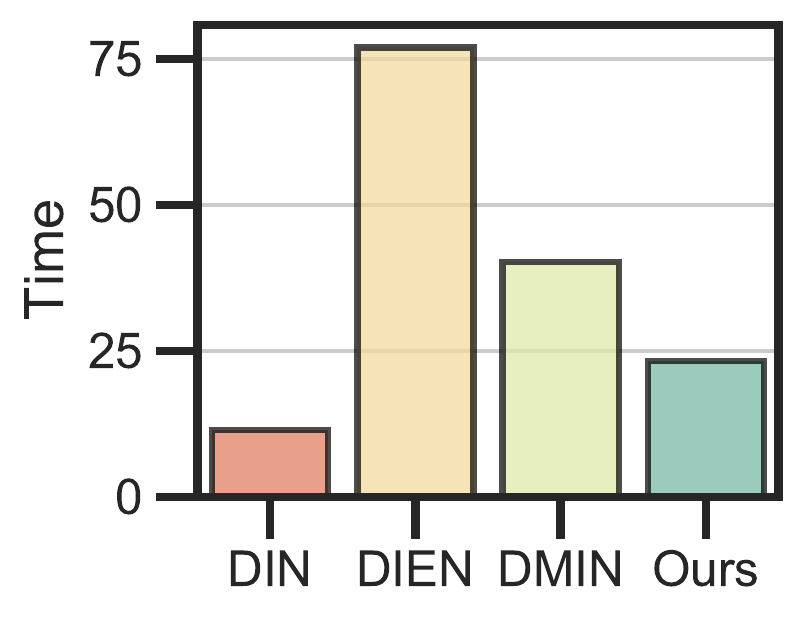}} 
	\subfloat[Training Time on Taobao]{
		\includegraphics[scale=0.43, trim=20 0 5 0, clip]{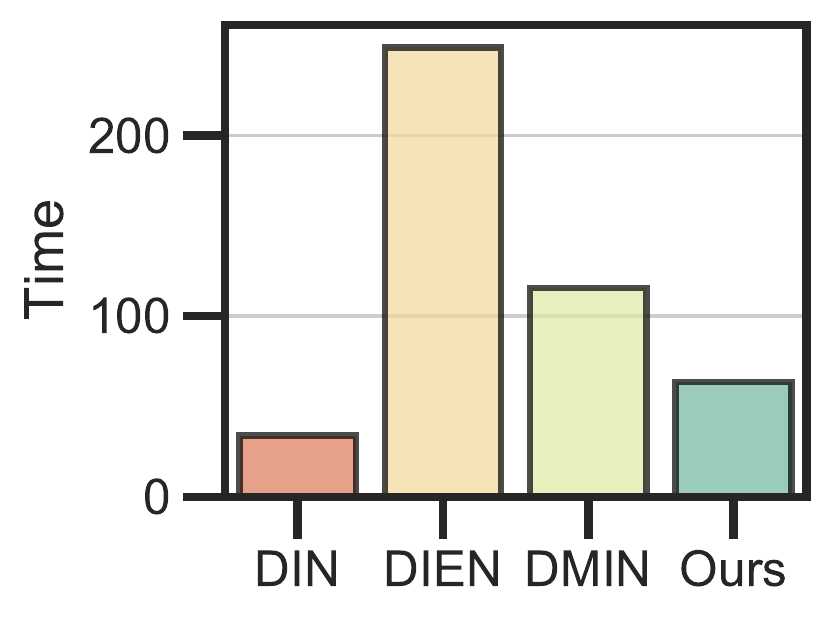}}
	\subfloat[Training Time on MovieLens]{
		\includegraphics[scale=0.43, trim=20 0 5 0, clip]{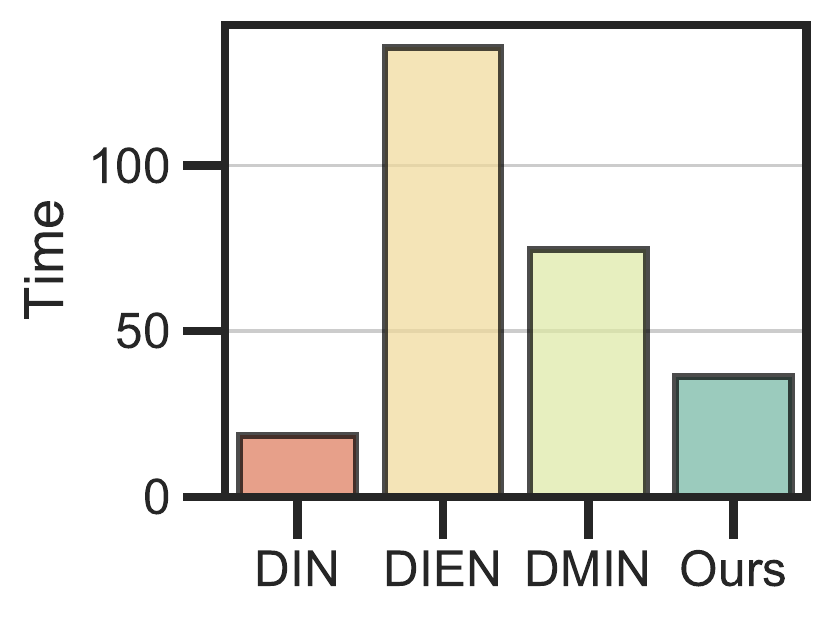}} \\
	\subfloat[Inference Time on Electronics]{
		\includegraphics[scale=0.43, trim=19 0 0 0, clip]{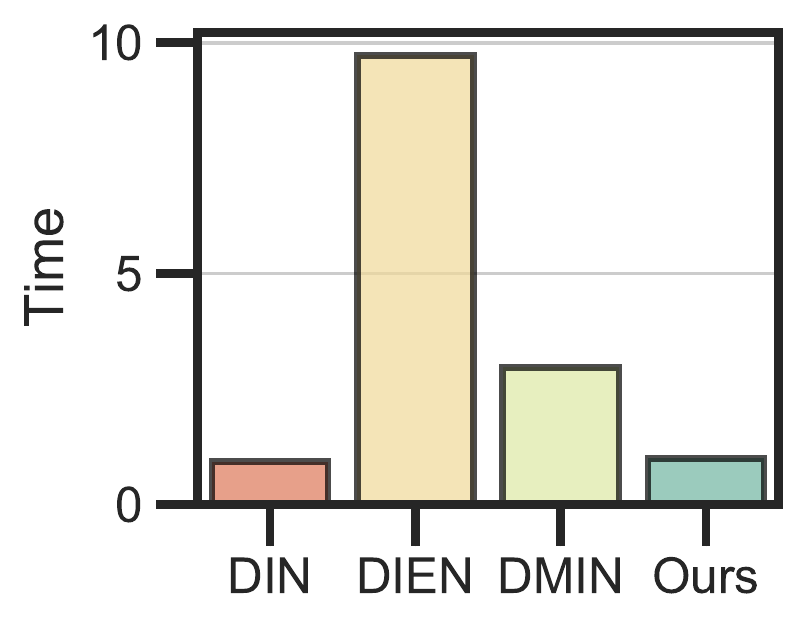}} 
	\subfloat[Inference Time on \\Books]{
		\includegraphics[scale=0.43, trim=19 0 0 0, clip]{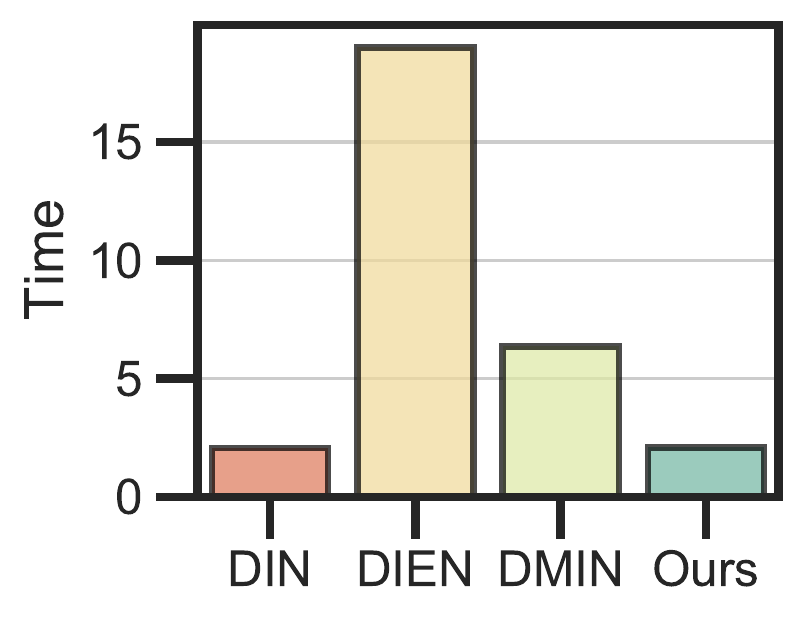}}
	\subfloat[Inference Time on Games]{
		\includegraphics[scale=0.43, trim=19 0 0 0, clip]{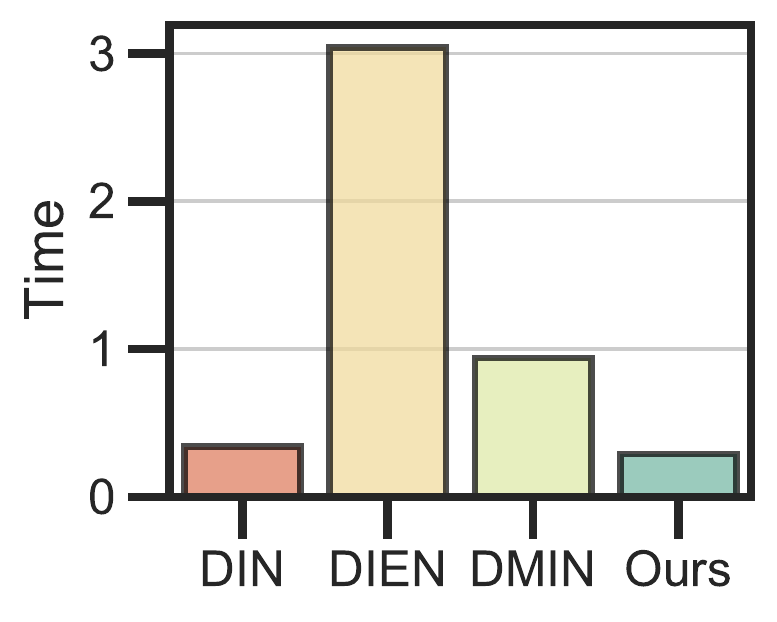}} 
	\subfloat[Inference Time on Taobao]{
		\includegraphics[scale=0.43, trim=19 0 0 0, clip]{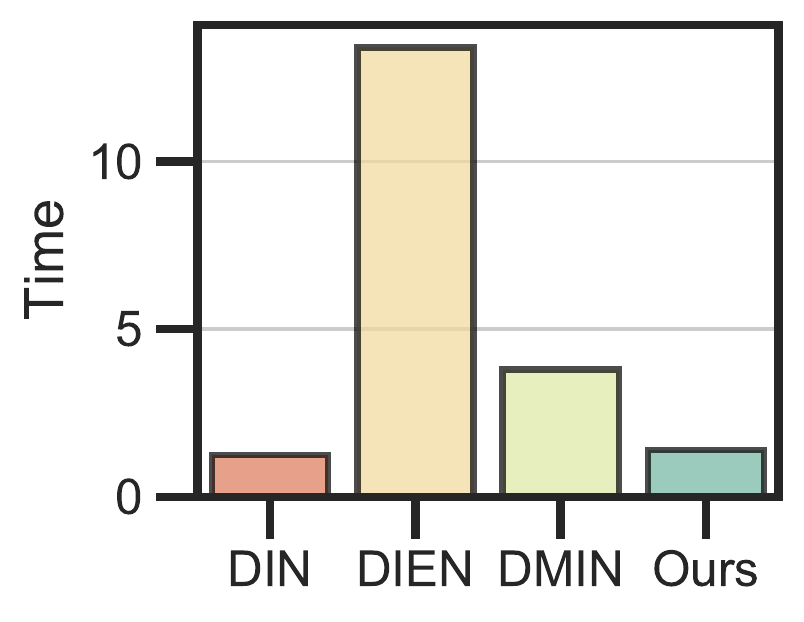}}
	\subfloat[Inference Time on MovieLens]{
		\includegraphics[scale=0.43, trim=19 0 0 0, clip]{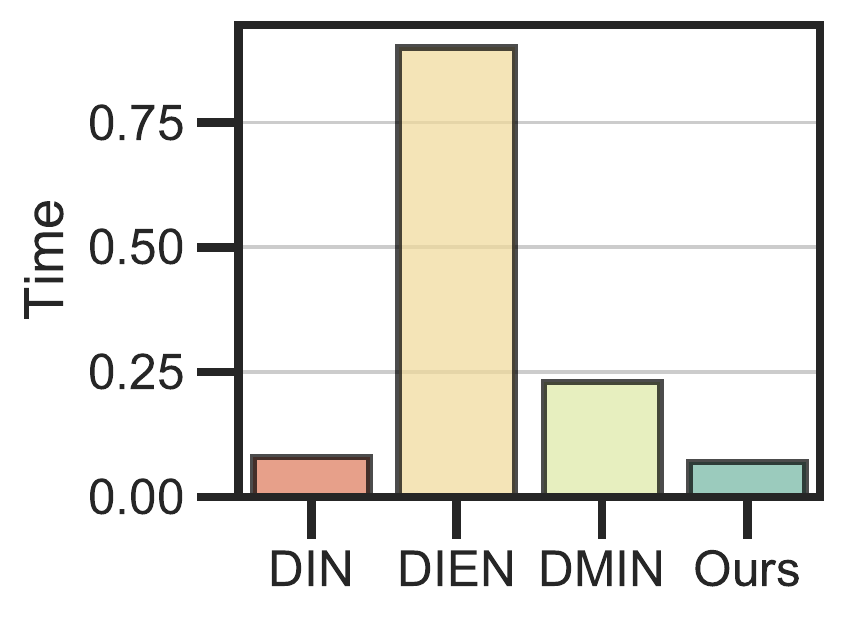}} \\
	\caption{Comparison of Training and Inference Time.} 
	\label{fig:speed}
\end{figure*}

\subsubsection{Evaluation Protocols}

In the field of recommender systems, we often use AUC as the metric. 
It could evaluate the ranking order of all the items with an estimated 
click-through rate:
\begin{equation}
	AUC=\frac{ \sum_{b\in\mathcal{S}_c}{\texttt{rank}(b)}-
		\frac{{N_{c}}\times({N_{c}}+1)}{2}}
	{{N_{c}} \times {N_{u}}},
\end{equation}
where $\mathcal{S}_c$ represents the set of all positive instances;
$b$ denotes a positive instance drawn from $\mathcal{S}_c$;
${N_{c}}$ denotes the number of the positive (clicked) instances; 
${N_{u}}$ is the number of the negative (unclicked) instances.
Moreover, we also provide the relative improvement over the base model, 
following \cite{groupctr}. 
Since the AUC score of a random model is about 0.5, 
the relative improvement $Impr$ in this paper is defined as the following:
\begin{equation}
	Impr= \left( \frac{ AUC^{(model)} - 0.5 }{AUC^{(base)} - 0.5} - 1 \right) \cdot 100 \%,
\end{equation}
where $AUC^{(model)}$ and $AUC^{(base)}$ refer to the AUC score w.r.t
the measured model and base model, respectively.

\subsection{Performance Comparison (RQ1)}
{ Tab.\ref{tab:perform_tab} shows the performance of all competitors 
on five public datasets. From this table, we can get the following 
observations:
\begin{itemize}
	\item All models with manual or automatic feature engineering outperform the 
	      base model significantly. This implies 
		  extensive feature engineering could improve the model performance on these 
		  datasets.
	\item Comparing the feature-crossing-based models, we see that the second-order feature-crossing methods 
		  (Wide\&Deep, PNN, DeepFM) have better performance than their high-order counterparts 
		  (xDeepFM and FGCNN) in the vast majority of the cases. 
		  The reason might be that higher-order interactions are not suitable
		  to model the behavioral features, though it is effective in dealing with general 
		  features.
	\item In most cases, general feature selection methods (AutoFIS and AutoGroup) are  comparable 
		  with second-order feature-crossing methods, but have worse performance than attention-based 
		  methods (DIN, DIEN, DMIN). To see this, attention mechanism could select different 
		  historical behaviors with respect to different target items. On the contrary, 
		  general feature selection methods can only provide the same feature subsets for 
		  all target items, without considering the relevance between users' past and present behaviors.
		  These results also reflect the difference between user behavior features and general
		  features.
	\item Meta-Wrapper achieves the best performance on all datasets. 
	      This implies that our meta-learned wrapping operator could enhance the learning of the 
	      feature selector (attentive module), thereby achieving better performance.
\end{itemize}}




\subsection{Ablation Study (RQ2)}

{To further justify our claim that improvements come from the 
meta-learned wrapping operator, we conduct a series of ablation experiments to validate the 
effectiveness of different modules of our method.

More specifically, we study the impact of adding 
the following components on top of a base model in the training process:
\begin{itemize}
	\item \textbf{C1}: An attentive module as the user interest selector;
	\item \textbf{C2}: The first term of Eq.(\ref{outer_loss_taylor_n})
	as an auxiliary loss used to regularize the feature selector;
	\item \textbf{C3}: The wrapper method described in 
	Example.(\ref{example_wrapper}) (i.e. GDmax-Wrapper) where we set
	inner-GD loops $N=1$;
	\item \textbf{C4}: The second term of Eq.(\ref{outer_loss_taylor_n})
	where we also set $N=1$ for fair comparisons.
\end{itemize}
Tab.\ref{tab:ablation_desc} further shows which components are 
included in the different compared methods (\textbf{M1-M4}).

The experimental results are presented in Fig.\ref{fig:albation}.
From these results, we have the following observations:
\begin{itemize}
	\item Comparing the performance of \textbf{M1} and the base model,
		  we can get that adding a feature selector (\textbf{C1}) on the top 
		  of the base predictor could significantly improve the model performance. 
		  Note that the performance can be found in the first row of 
		  Tab.\ref{tab:perform_tab}, which is not repeatedly shown in 
		  Fig.\ref{fig:albation}.
	\item \textbf{M2} brings an improvement on 4 out of 5 datasets with respect to \textbf{M1}. This illustrates that adding the first term of Eq.(\ref{outer_loss_taylor_n}) 
		  (\textbf{C2}) actually works in CTR prediction tasks.
	\item \textbf{M3} has lower performance than \textbf{M2} on all datasets, which means that
		  (\textbf{C3}) would result in performance degradation. The reason might be that inner 
		  parameters cannot be sufficiently trained when ${N}$ is small. Thus the inner solution 
		  is not precise enough to direct the outer-level problem, leading to the 
		  ineffectiveness of the wrapper method.
	\item \textbf{M4} achieves the best performance on all datasets, which manifests the second term 
	      of Eq.(\ref{outer_loss_taylor_n}) (\textbf{C4}) is also a source of performance improvements.
		  This result is consistent with our theoretical analysis in Sec.(\ref{sec:theory}).
\end{itemize}
}

\begin{table}[htbp]
	\centering
	\caption{Methods (\textbf{M1}-\textbf{M4}) w.r.t 
			 Components (\textbf{C1}-\textbf{C4}).}
	\begin{tabular*}{\hsize}{@{}@{\extracolsep{\fill}}lcccccc@{}}
		\multicolumn{1}{c}{\multirow{2}[2]{*}{\tabincell{c}{}}} & 
		\multicolumn{1}{c}{\multirow{2}[2]{*}{\tabincell{c}{\textbf{C1}}}} & 
		\multicolumn{1}{c}{\multirow{2}[2]{*}{{\tabincell{c}{\textbf{C2}}}}} & 
		\multicolumn{1}{c}{\multirow{2}[2]{*}{\tabincell{c}{\textbf{C3}}}} & 
		\multicolumn{1}{c}{\multirow{2}[2]{*}{\tabincell{c}{\textbf{C4}}}} \\
		&       &       &       &  \\
		\midrule
		\multirow{2}[1]{*}{\textbf{M1}} & 
		\multirow{2}[1]{*}{\large \checkmark} & 
		\multirow{2}[1]{*}{\large $\times$} & 
		\multirow{2}[1]{*}{\large $\times$} & 
		\multirow{2}[1]{*}{\large $\times$} \\
		&       &       &       &  \\
		\multirow{2}[1]{*}{\textbf{M2}} & 
		\multirow{2}[0]{*}{\large \checkmark} & 
		\multirow{2}[0]{*}{\large \checkmark} & 
		\multirow{2}[0]{*}{\large $\times$} & 
		\multirow{2}[0]{*}{\large $\times$} \\
		&       &       &       &  \\
		\multirow{2}[1]{*}{\textbf{M3}} & 
		\multirow{2}[0]{*}{\large \checkmark} & 
		\multirow{2}[0]{*}{\large \checkmark} & 
		\multirow{2}[0]{*}{\large \checkmark} & 
		\multirow{2}[0]{*}{\large $\times$} \\
		&       &       &       &  \\
		\multirow{2}[1]{*}{\textbf{M4}} & 
		\multirow{2}[1]{*}{\large \checkmark} & 
		\multirow{2}[1]{*}{\large \checkmark} & 
		\multirow{2}[1]{*}{\large \checkmark} & 
		\multirow{2}[1]{*}{\large \checkmark} \\
		&       &       &       &  \\
		\bottomrule
	\end{tabular*}%
	\label{tab:ablation_desc}%
\end{table}%

\begin{figure}[htbp]
	\centering
	\subfloat[$\mu$ w.r.t Electronics]{
		\includegraphics[scale=0.42, trim=20 0 0 5, clip]{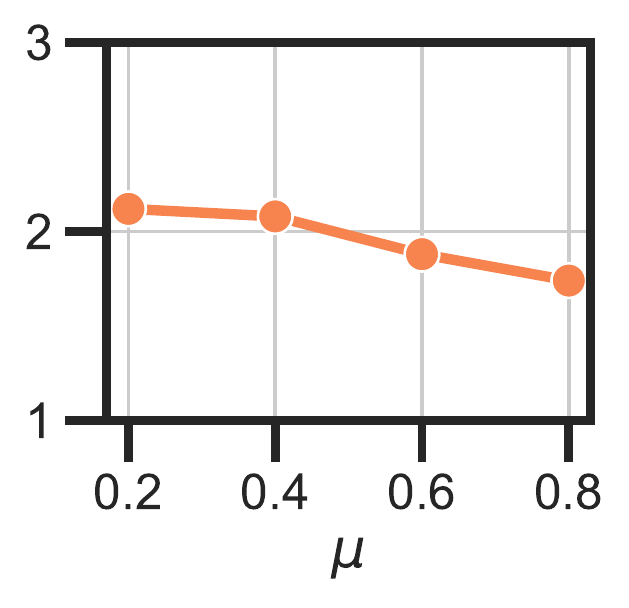}} 
	\subfloat[$N$ w.r.t Electronics]{
		\includegraphics[scale=0.42, trim=20 0 0 5, clip]{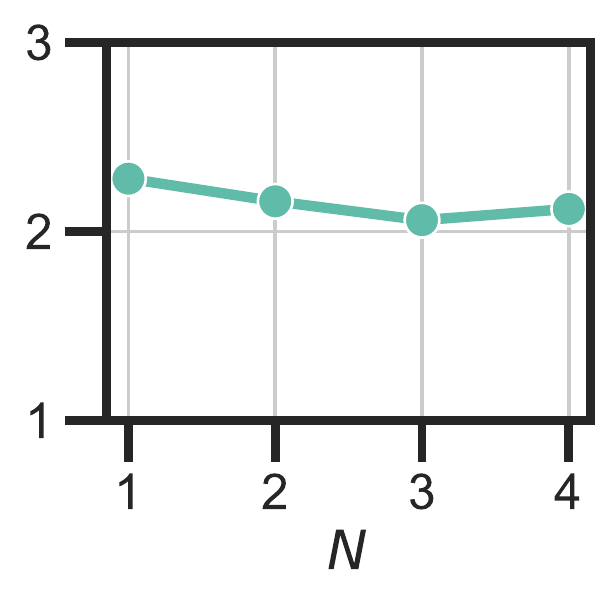}} 
	\subfloat[$\beta$ w.r.t Electronics]{
		\includegraphics[scale=0.42, trim=20 0 0 5, clip]{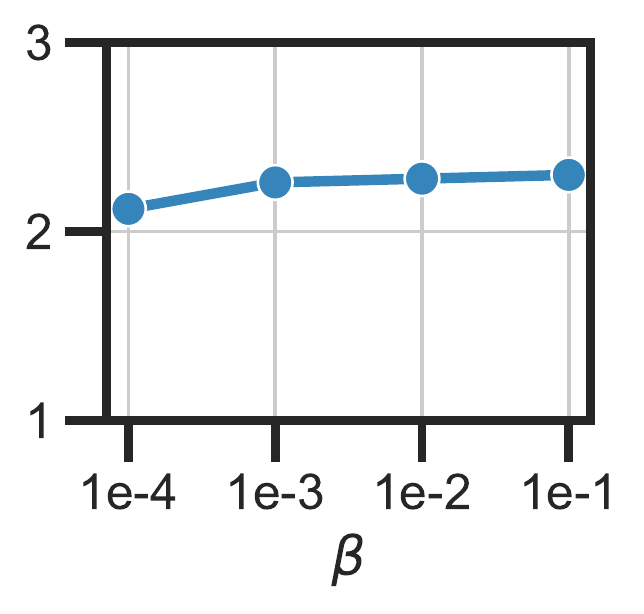}} \\
	\subfloat[$\mu$ w.r.t Books]{
		\includegraphics[scale=0.42, trim=20 0 0 5, clip]{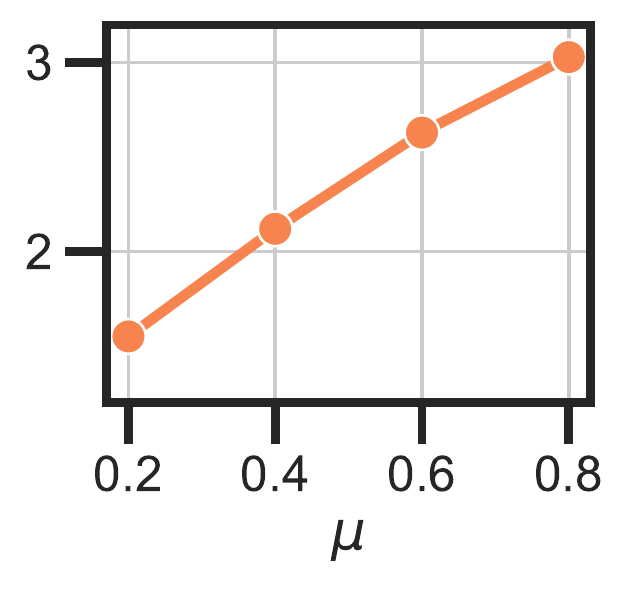}} 
	\subfloat[$N$ w.r.t Books]{
		\includegraphics[scale=0.42, trim=20 0 0 5, clip]{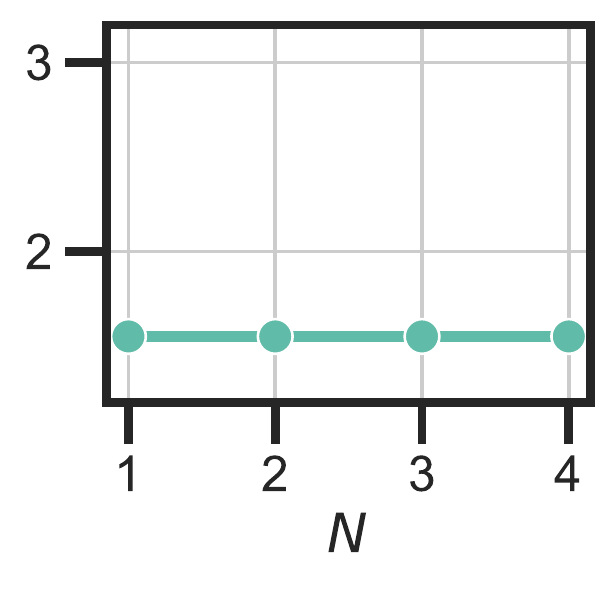}} 
	\subfloat[$\beta$ w.r.t Books]{
		\includegraphics[scale=0.42, trim=20 0 0 5, clip]{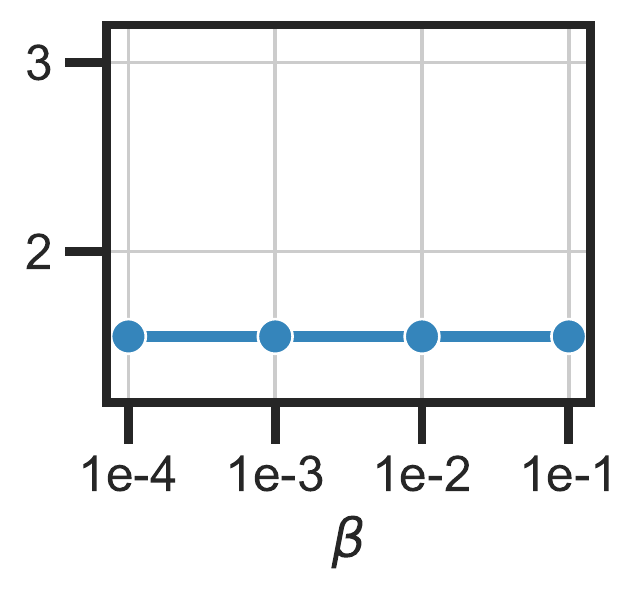}} \\
	\subfloat[$\mu$ w.r.t Games]{
		\includegraphics[scale=0.42, trim=20 0 0 5, clip]{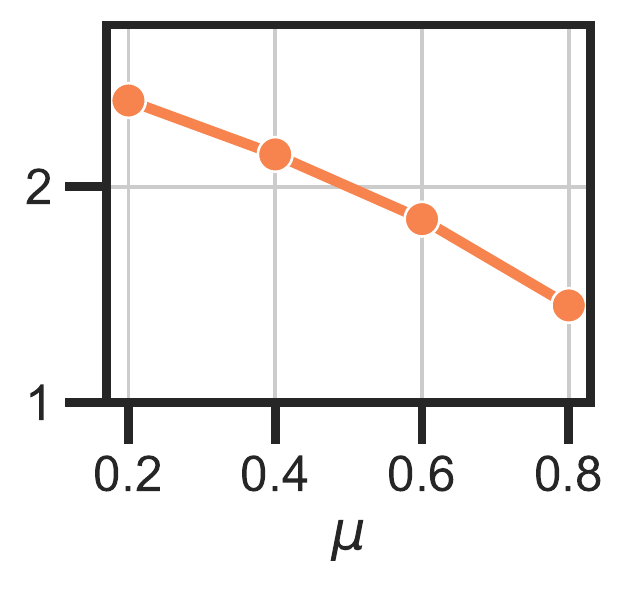}} 
	\subfloat[$N$ w.r.t Games]{
		\includegraphics[scale=0.42, trim=20 0 0 5, clip]{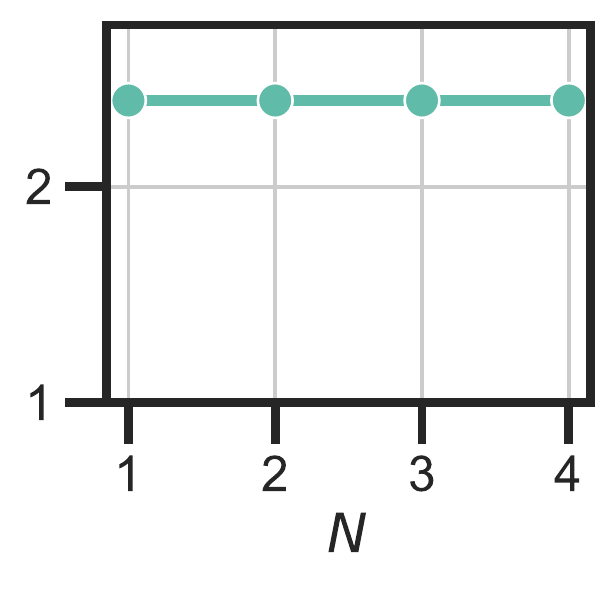}}
	\subfloat[$\beta$ w.r.t Games]{
		\includegraphics[scale=0.42, trim=20 0 0 5, clip]{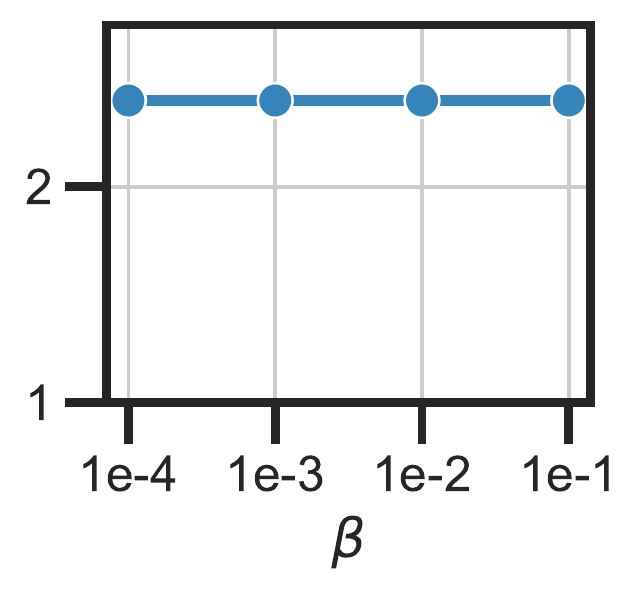}} \\
	\subfloat[$\mu$ w.r.t Taobao]{
		\includegraphics[scale=0.42, trim=20 0 0 5, clip]{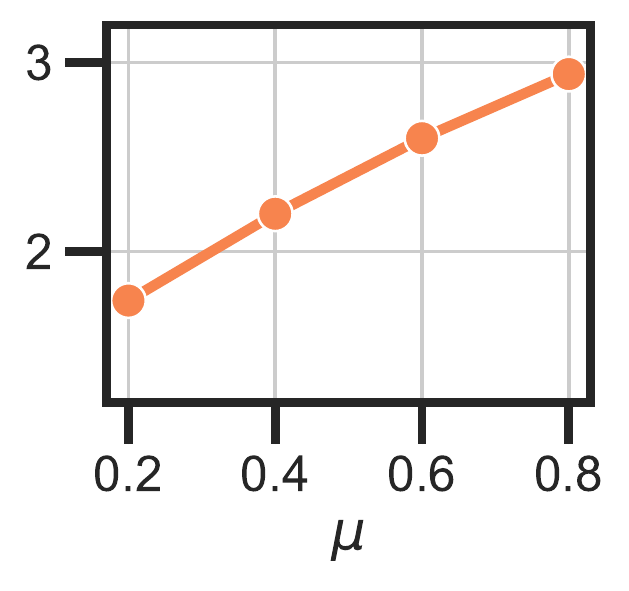}} 
	\subfloat[$N$ w.r.t Taobao]{
		\includegraphics[scale=0.42, trim=20 0 0 5, clip]{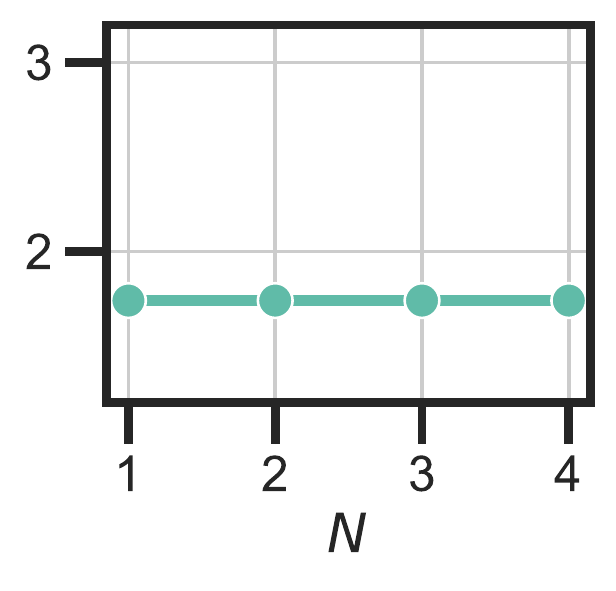}}
	\subfloat[$\beta$ w.r.t Taobao]{
		\includegraphics[scale=0.42, trim=20 0 0 5, clip]{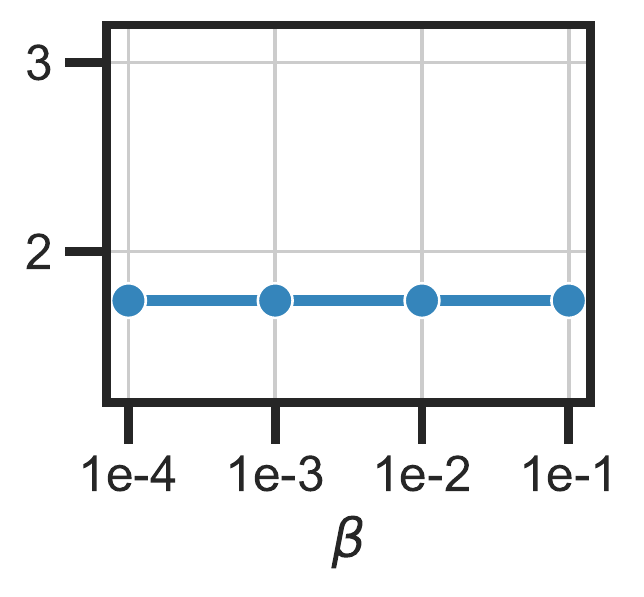}} \\
	\subfloat[$\mu$ w.r.t MovieLens]{
		\includegraphics[scale=0.42, trim=20 0 0 5, clip]{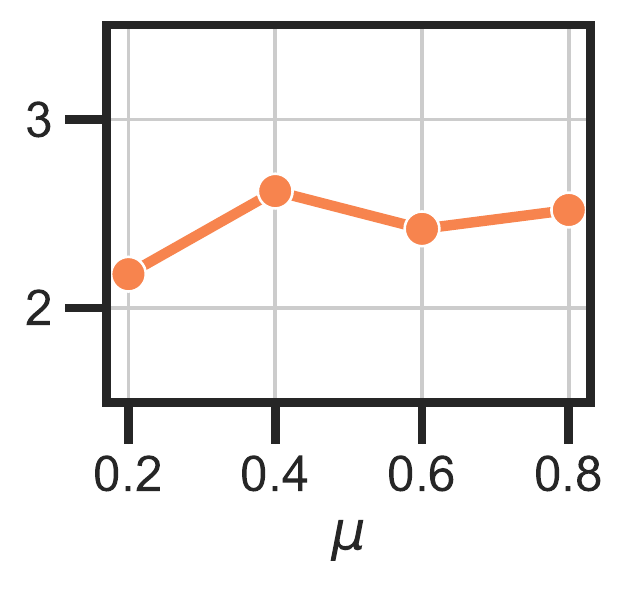}} 
	\subfloat[$N$ w.r.t MovieLens]{
		\includegraphics[scale=0.42, trim=20 0 0 5, clip]{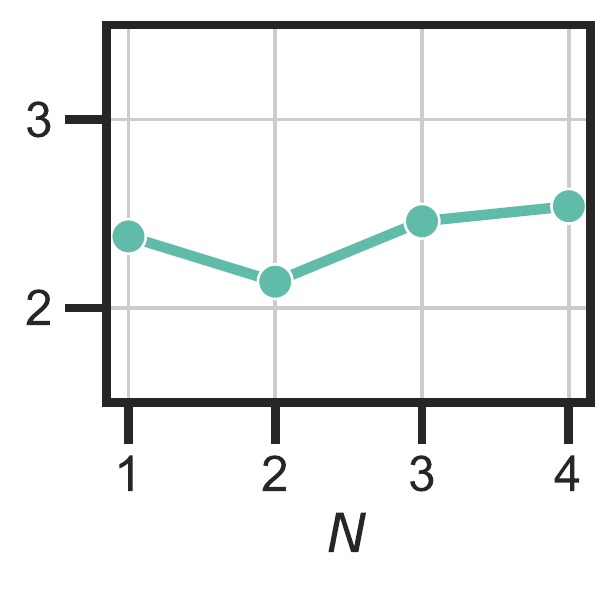}} 
	\subfloat[$\beta$ w.r.t MovieLens]{
		\includegraphics[scale=0.42, trim=20 0 0 5, clip]{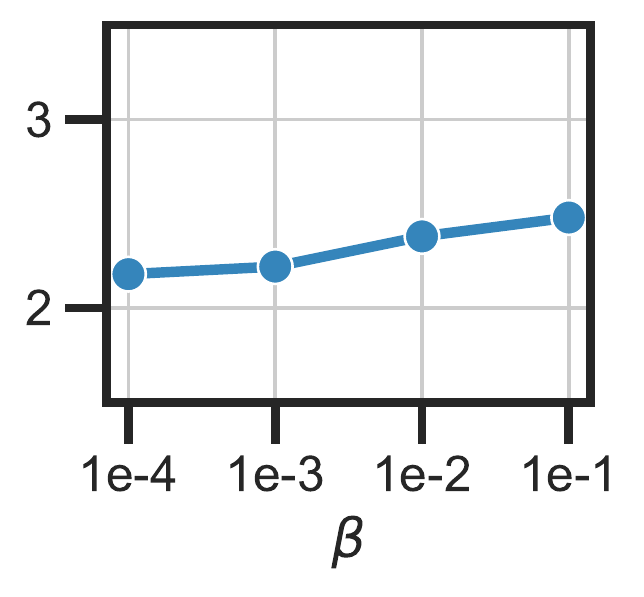}} \\
	\caption{ Hyperparameter Study on Public Datasets.}
	\label{fig:hyper}
\end{figure}

\subsection{Generalization Performance (RQ3)}

{
In this section, we conduct  experiments on the synthetic
data to demonstrate the overfitting in CTR prediction. 
We generate the synthetic data with user behavioral noises
and observe to what extent are different models overfit.
Meanwhile, we also show that Meta-Wrapper can better deal with 
such a situation.

\subsubsection{Synthetic Data} 
We generate $100$ users and $10000$ items, while users and items are both 
divided into 10 groups. Here each group represents a type of interest.
More specifically, the $i$-th user is assigned to the $g_u(i)$-th group,
where $g_u(i)=(i\ MOD\ 10)$. Similarly, the $j$-th item is assigned to the $g_i(j)$-th group, 
where $g_i(j)=(j\ MOD\ 1000)$.
For each user $i$, we select $50$ items as his/her interested items (i.e. positive instances),
where the items are uniformly sampled from the subset 
$\{j | g_i(j)=g_u(i), j \in \mathcal{N}, 1 \leq j \leq 10000\}$.
Meanwhile, we uniformly sample $50$ non-interested items (i.e. negative instances) for each user 
$i$ from the item subset 
$\{j | g_i(j) \neq g_u(i),  j \in \mathcal{N}, 1 \leq j \leq 10000\}$.
In this way, we can get a noiseless dataset with $5000$ positive instances and 
$5000$ negative instances.
To simulate the noises, we now apply a perturbation to this dataset.
More specifically, we generate labels $\bm{Y_p}$ for positive instances 
such that $\bm{Y_p} \sim Binomial(5000, 0.5)$, and 
labels $\bm{Y_n}$ for negative instances such that 
$\bm{Y_n} \sim Binomial(5000, 0.2)$.
This implies that 1) there are about $50\%$ interested items labeled
as non-clicked, and 2) about $20\%$ non-interested items labeled
as clicked.
In real-word scenario, it is possible that some interested items 
are ignored. 
Meanwhile, non-interested items might also
be clicked by mistake. 
Thus we can say that such perturbations do exist
in real-world CTR predictions, though it is not necessarily severe as
our simulation. The reason why we construct such strong noises is to
clearly show the overfitting and the improvements of Meta-Wrapper.

\subsubsection{Experimental Results on Synthetic Data}
Based on the synthetic data, we compare our proposed method with the 
1) \textbf{Base:} the base model;
2) \textbf{FS:} the base model with attentive feature selector (FS);
3) \textbf{FS+L2:} the base model with FS and $L2$ weight decay ($L2$);
4) \textbf{FS+DO:} the base model with FS and $L2$ dropout (DO).

\begin{figure}[htbp]
	\centering
	\includegraphics[scale=0.55, trim=9 0 6 0, clip]{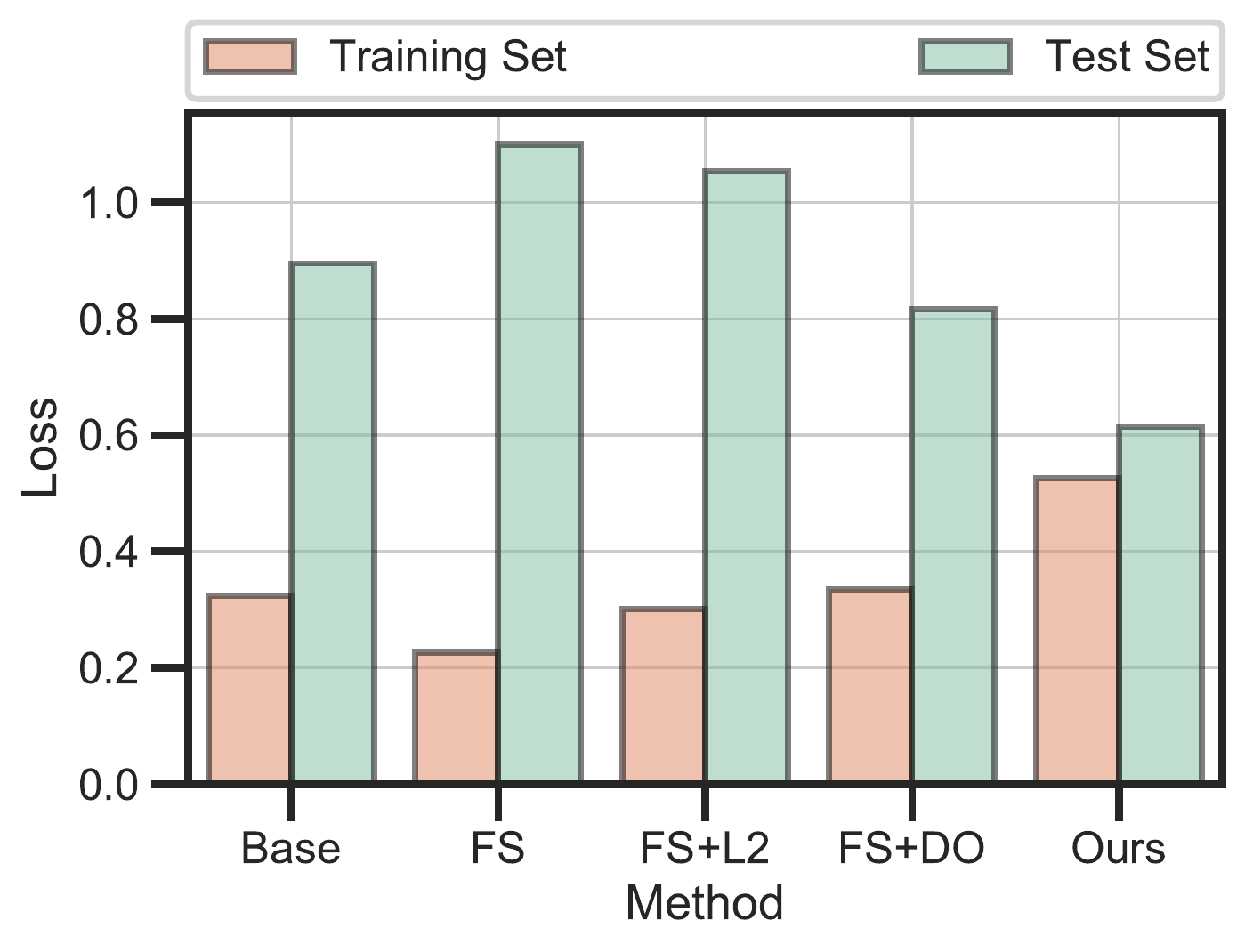}
	\caption{Overfitting on Synthetic Data With Strong Noises.} 
	\label{fig:sim_data}
\end{figure}

The experimental results are shown in Fig.\ref{fig:sim_data}. 
For each approach, we can observe a major gap between the training and 
test loss, which results from our strong noises. Notably, the \textbf{FS}
has lower training loss and higher test loss compared with \textbf{Base}.
This means that attentive feature selector might exacerbate overfitting.
Meanwhile, \textbf{FS+L2} and  \textbf{FS+DO} are a little better than \textbf{FS}, but its
overfitting is still more serious than \textbf{Base}. This demonstrates
the impact of $L2$ weight decay is limited in this scene. Finally, 
\textbf{MW} significantly reduces the training and test loss compared
with other methods, where the overfitting is even less serious than \textbf{Base}.

From these results, we can conclude that 1) attentive feature selector
would increase the overfitting risk, and 2) Meta-Wrapper is an effective
approach to alleviate overfitting.
}


\subsection{Efficiency Analysis (RQ4)}
 
{
The data volume tends to be extremely large in practical applications, 
but the available computing resources are always limited.
In such a scenario, a CTR prediction model should not only pay attention
to the accuracy, but also consider the efficiency.
Therefore, despite the effectiveness of Meta-Wrapper has been 
validated, it is also necessary to empirically analyse its computational 
efficiency.

In this section, we compare the training and inference time of Meta-Wrapper
with three industrial models: DIN, DIEN and DMIN on one NVIDIA TITAN RTX GPU.
In these experiments, we fix the inner loops $N=1$ for Meta-Wrapper.
The experimental results are demonstrated in Fig.\ref{fig:speed}. 
From this figure, we can observe that:
\begin{itemize}
	\item At the training stage, DIN is the fastest approach. This result
	is in line with our expectation, since DIN is the most simple method
	compared with the other four competitors. Meanwhile, DIEN and DMIN need
	more training time, especially DIEN. This is because they have the 
	complex modeling of sequential evolution. DMIN is faster than DIEN since 
	it uses transformer instead of RNN to improve the concurrency. Moreover,
	the training time of Meta-Wrapper is about twice as DIN, which  
	is consistent with our theoretical analysis.
	\item At the inference stage, DIEN is still the slowest method, and DMIN
	is a little faster than DIEN by virtue of its parallel computing 
	capability. Meanwhile, DIN and Meta-Wrapper are significantly 
	faster than others. Notably, our proposed meta-learning strategy only
	focuses on the training stage, so the inference time of Meta-wrapper 
	should be equal to DIN. This is also proved by experimental results.
\end{itemize}

\subsection{Hyperparameter Study (RQ5)}

Recall that our final loss function Eq.(\ref{a_plus_b}) introduces 
three hyperparameters: 
1) the weight of feature selection task $\mu$, 
2) inner-level learning rate $\beta$ and 
3) the number of inner-loops ${N}$.
We study how these hyperparameters affect the model performance.

Fig.(\ref{fig:hyper}) presents our experimental results on 
five datasets, respectively. 
On Electronics dataset, all three hyperparameters do not have a great impact 
on performance. On Books, Games and Taobao, 
$\mu$ has a relatively larger impact on performance, while 
the impact of $\beta$ and ${N}$ is almost negligible.
On MovieLens dataset, 
$\mu$, $\beta$ and ${N}$ all have an impact on performance, but  the impact is quite small.
Above all, the changing of performance is relatively evident when tuning the value of $\mu$, 
while it is quite stable across different $\beta$ and ${N}$.
Therefore, we should give priority to tune the value of $\mu$ in 
practical applications, for the sake of better performance. Moreover,
the stable performance w.r.t ${N}$ shows that Meta-Wrapper is 
robust towards the number of inner loops. In this way, the wrapper 
method could become computationally tractable by setting a smaller
number of inner-loops ${N}$.


}
\section{Conclusion}
In this paper, we propose a differentiable wrapper method, namely Meta-Wrapper, 
to automatically select user interests in CTR prediction tasks. 
More specifically, we first regard user interest modeling as a feature 
selection problem. Then we provide the objective function, under the 
framework of the wrapper method, for such a problem. Moreover, we design a 
differentiable wrapping operator to simultaneously improve the efficiency and 
flexibility of the wrapper method. Based on these, the learning problems of the 
feature selector and downstream predictor are unified as a bilevel optimization 
which can be solved by a meta-learning algorithm. Meanwhile, we prove the 
effectiveness of our proposed method from a theoretical perspective. From 
the experimental results on three datasets, we also observe that our method 
significantly improves the performance of CTR prediction in recommender systems. 
In the future, we will try to incorporate more criteria besides performance 
into the framework of the differentiable wrapper to improve the user 
interest selection in CTR prediction.

\section*{Acknowledgment}
This work was supported in part by the National Key R\&D Program 
of China under Grant 2018AAA0102003, in part by National Natural Science 
Foundation of China: 61931008, 61620106009, 61836002 and 61976202, 
in part by the Fundamental Research Funds for the Central Universities, 
in part by Youth Innovation Promotion Association CAS, in part by 
the Strategic Priority Research Program of Chinese Academy of Sciences, 
Grant No. XDB28000000, and in part by National Postdoctoral Program 
for Innovative Talents under Grant No. BX2021298.


%

\ifCLASSOPTIONcaptionsoff
\newpage
\fi



%

%

\bibliographystyle{IEEEtran}
\bibliography{ref}{}





\vspace{0cm}\begin{IEEEbiography}[{\includegraphics[width=1in,height=1.25in,clip,keepaspectratio]{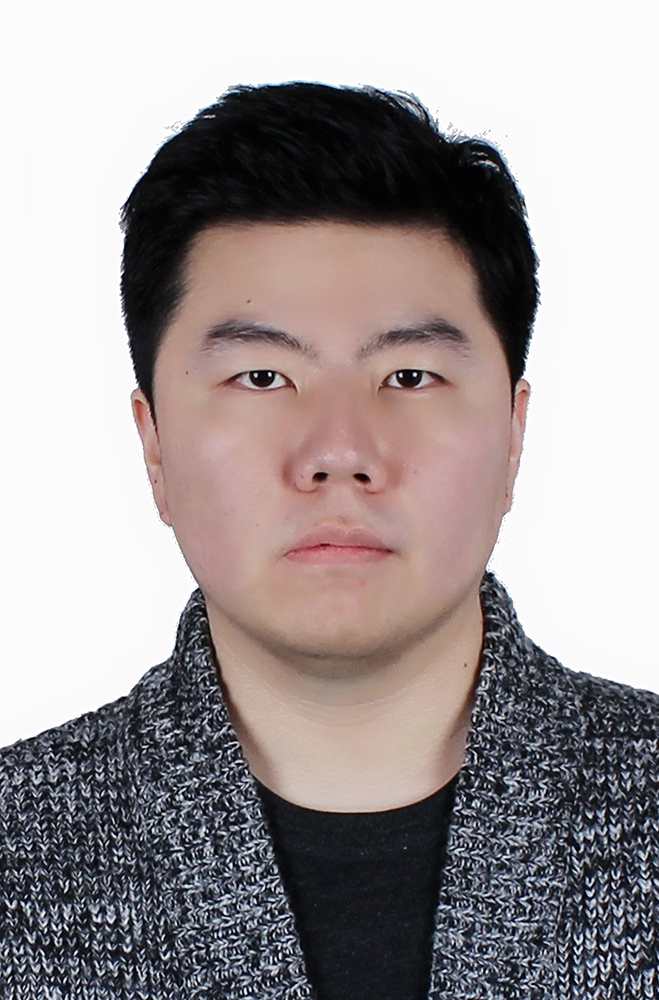}}]{Tianwei Cao}
received the M.E. degree in computer science and technology from University of Science and Technology Beijing (USTB) in 2017. He is currently pursuing the Ph.D. degree with University of Chinese Academy of Sciences. His research interests is data mining, especially recommender system and computational advertising.
\end{IEEEbiography}


\vspace{0cm}\begin{IEEEbiography}[{\includegraphics[width=1in,height=1.25in,clip,keepaspectratio]{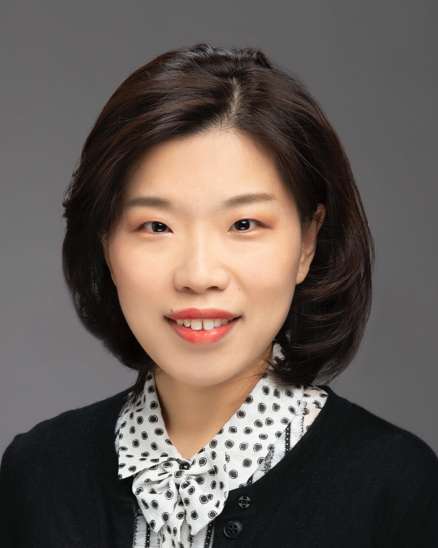}}]{Qianqian Xu}
	received the B.S. degree in computer
	science from China University of Mining
	and Technology in 2007 and the Ph.D. degree
	in computer science from University of Chinese
	Academy of Sciences in 2013. She is currently
	an Associate Professor with the Institute of Computing
	Technology, Chinese Academy of Sciences,
	Beijing, China. Her research interests include
	statistical machine learning, with applications
	in multimedia and computer vision. She has
	authored or coauthored 40+ academic papers in
	prestigious international journals and conferences, 
	including T-PAMI/IJCV/T-IP/T-KDE/ICML/NeurIPS/CVPR/AAAI, {etc}. 
	She served as a reviewer for several top-tier journals and conferences, 
	such as T-PAMI, ICML, NeurIPS, ICLR, CVPR, ECCV, AAAI, IJCAI, and ACM MM, etc.
\end{IEEEbiography}

\vspace{0cm}\begin{IEEEbiography}[{\includegraphics[width=1in,height=1.25in,clip,keepaspectratio]{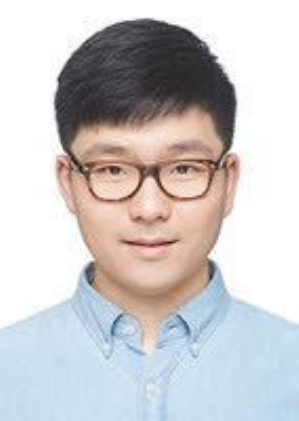}}]{Zhiyong Yang} received the M.Sc. degree in
	computer science and technology from University of Science and Technology Beijing
	(USTB) in 2017, and the Ph.D. degree from
	University of Chinese Academy of Sciences
	(UCAS) in 2021. He is currently a postdoctoral research fellow with the University of
	Chinese Academy of Sciences. His research
	interests lie in machine learning and learning theory, with special focus on AUC optimization, meta-learning/multi-task learning,
	and learning theory for recommender systems. He has authored
	or coauthored 20+ academic papers in top-tier international
	conferences and journals including T-PAMI/ICML/NeurIPS/CVPR.
	He served as a reviewer for several top-tier journals and conferences such as T-PAMI, ICML, NeurIPS and ICLR.
\end{IEEEbiography}

\vspace{0cm}\begin{IEEEbiography}[{\includegraphics[width=1in,height=1.25in,clip,keepaspectratio]{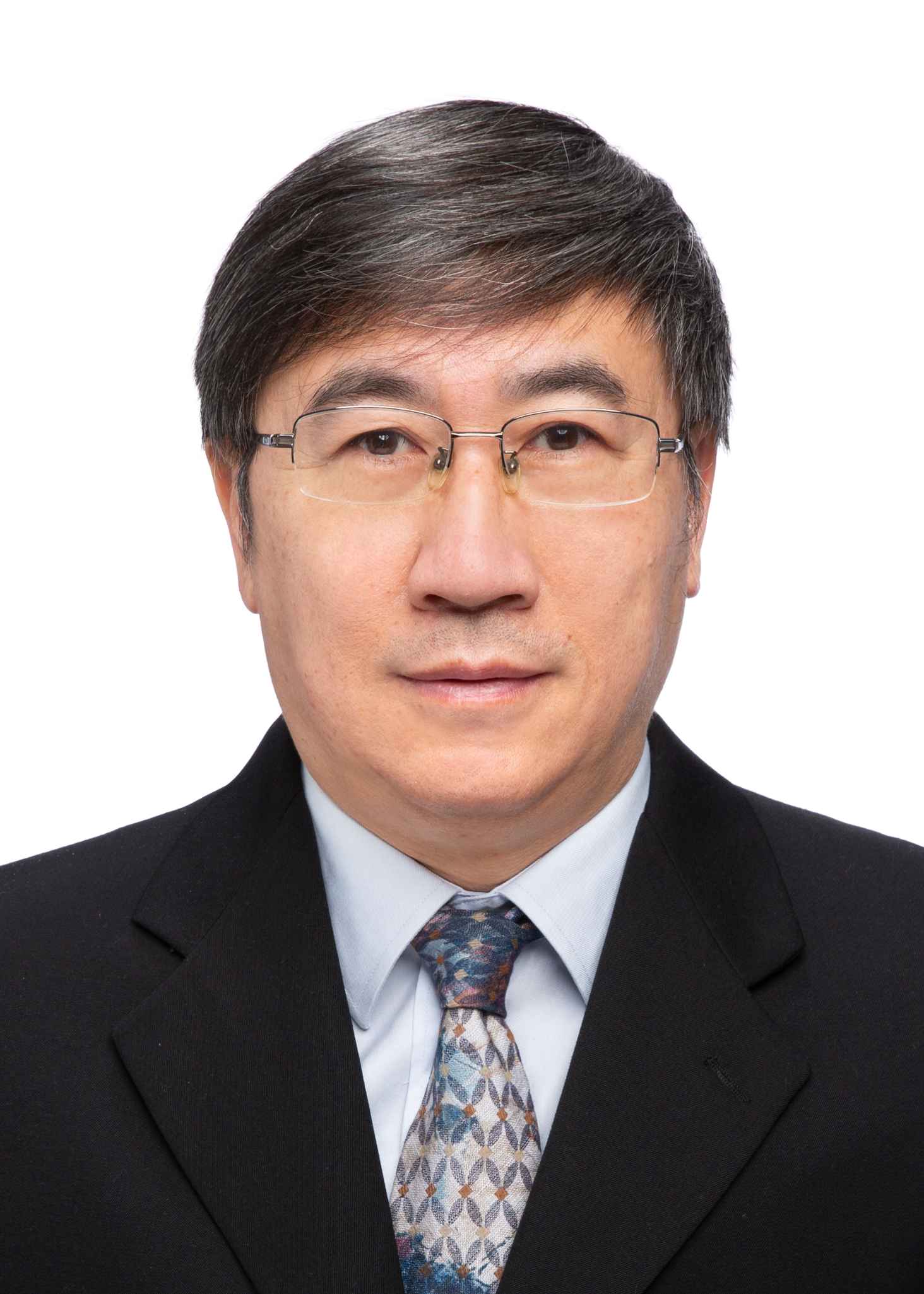}}]{Qingming Huang} 
is a chair professor in the University of Chinese Academy of Sciences and an adjunct research professor in the Institute of Computing Technology, Chinese Academy of Sciences. He graduated with a Bachelor degree in Computer Science in 1988 and Ph.D. degree in Computer Engineering in 1994, both from Harbin Institute of Technology, China. His research areas include multimedia computing, image processing, computer vision and pattern recognition. He has authored or coauthored more than 400 academic papers in prestigious international journals and top-level international conferences. He is the associate editor of IEEE Trans. on CSVT and Acta Automatica Sinica, and the reviewer of various international journals including IEEE Trans. on PAMI, IEEE Trans. on Image Processing, IEEE Trans. on Multimedia, etc. He is a Fellow of IEEE and has served as general chair, program chair, track chair and TPC member for various conferences, including ACM Multimedia, CVPR, ICCV, ICME, ICMR, PCM, BigMM, PSIVT, etc.
\end{IEEEbiography}

\clearpage
\appendix[Convergence Result of Alg.\ref{training}]\begin{table*}[htbp]
	\renewcommand\arraystretch{1.5} 
	\normalsize
	\centering
	\caption{Notations and Descreptions.}
	\begin{tabular*}{\hsize}{ll}
	Notations & \multicolumn{1}{c}{Descreptions} \\
	\toprule
	$N_\theta$, $N_\phi$ & 
		$N_\theta$ and $N_\phi$ are dimensions of parameters vector $\bm{\theta}$ 
		and $\bm{\phi}$, and 
		$p=N_\theta+N_\phi$; \\
	$p$ & The number of all parameters: $p=N_\theta + N_\phi$; \\
	${\bm{\theta}}, {\bm{\theta}_w}, {\bm{\theta}_u}$ & 
	Parameters value of base predictor: 
	${\bm{\theta}}, {\bm{\theta}_w}, {\bm{\theta}_u}\in \mathbb{R}^{N_\theta}$; \\
	${\bm{\phi}}, {\bm{\phi}_w}, {\bm{\phi}_u}$ & 
	Parameters value of feature selector:
	${\bm{\theta}}, {\bm{\phi}_w}, {\bm{\phi}_u} \in \mathbb{R}^{N_\phi}$; \\
	$\bm{z}$, $\bm{z}_{w}$, $\bm{z}_{u}$ & 
	All parameters:
	$\bm{z}=[\bm{\theta}^\top,\bm{\phi}^\top]^\top \in \mathbb{R}^p$,
	$\bm{z}_{w}=[{\bm{\theta}^\top_w},{\bm{\phi}^\top_q}]^\top \in \mathbb{R}^p$,
	$\bm{z}_{u}=[{\bm{\theta}^\top_u},{\bm{\phi}^\top_u}]^\top \in \mathbb{R}^p$; \\
	$\bm{\theta}^{(i)}(\bm{\phi})$ & 
	The value of $\bm{\theta}$ after $i$ loops of inner-GD;\\
	$\bm{z}_{w}^{{(i})}$, $\bm{z}_{u}^{{(i})}$ & 
	Parameters value after $i$ loops of inner-GD:
	$\bm{z}_{w}^{{(i})}=
	[{(\bm{\theta}^{(i)}(\bm{\phi}_w))}^\top ,{\bm{\phi}_w}^\top]^\top$,
	$\bm{z}_{u}^{{(i})}=
	[{(\bm{\theta}^{(i)}(\bm{\phi}_u))}^\top ,{\bm{\phi}_u}^\top]^\top$; \\
	$G_{w}^{(i,j)}, G_{u}^{(i,j)}$ & 
	Gradient on $\mathcal{D}^{(in)}$ across inner-GD loops: 
	$G_{w}^{(i,j)}=\nabla_{\bm{z}_{w}^{{(i})}}              
	L^{(in)}(\bm{z}_{w}^{{(j})}), 
	G_{u}^{(i,j)}=\nabla_{\bm{z}_{u}^{{(i})}}              
	L^{(in)}(\bm{z}_{u}^{{(j})})$;\\
	$\hat{G}_{w}^{(i,j)}, \hat{G}_{u}^{(i,j)}$ & 
	Gradient on $\mathcal{D}^{(out)}$ across inner-GD loops: 
	$G_{w}^{(i,j)}=\nabla_{\bm{z}_{w}^{{(i})}}              
	L^{(out)}(\bm{z}_{w}^{{(j})}), 
	G_{u}^{(i,j)}=\nabla_{\bm{z}_{u}^{{(i})}}              
	L^{(out)}(\bm{z}_{u}^{{(j})})$;\\
	$\mathcal{G}_{\square}$ & 
	The outer-level gradient: $\nabla_{\square}{\mathcal{L}(\bm{z})}$,
	where $\square \in \{ \bm{z}, \bm{\phi}, \bm{\theta}\}$
	and $\mathcal{L}$ is defined in Eq.(\ref{a_plus_b});\\
	$H_{w}^{(i)}, H_{u}^{(i)}$ & 
	Hessian at $j$-th inner-GD step:
	$\nabla^{2}_{\bm{z}_{w}^{{(i})}}            
	L^{(in)}(\bm{z}_{w}^{{(i})})^\top$,
	$\nabla^{2}_{\bm{z}_{u}^{{(i})}}            
	L^{(in)}(\bm{z}_{u}^{{(i})})^\top$; \\
	$\mathcal{T}$ & 
	Meta-learning Task: 
	$\mathcal{T}_i=\{\mathcal{D}^{(in)}_i, \mathcal{D}^{(out)}_i\}$,
	where subscript $i$ is the number of outer-GD step; \\
	$\mathcal{P}(\mathcal{T})$ & 
	The distribution of $\mathcal{T}$;\\ 
	$\bm{z}_{(i)}$ & 
	Parameters value after $i$ loops of Alg.\ref{training} (outer-GD). \\
	$\mathcal{F}$ & 
	Expectated loss w.r.t $\mathcal{P}(\mathcal{T})$:
	$
	\mathbb{E}(                                           
		{\mathcal{L}(\bm{z})})$, ; \\
	$\mathcal{F}_i, \mathcal{F}(\bm{z}_{(i)})$ & 
	Expectated loss w.r.t $\mathcal{P}(\mathcal{T})$ at $i$-th outer-GD step:
	$\mathcal{F}_i=\mathcal{F}(\bm{z}_{(i)})=\mathbb{E}(                                         
		{\mathcal{L}(\bm{z}_{(i)})})$; \\
	\bottomrule
	\end{tabular*}%
	\label{tab:symbols}%
\end{table*}%

We start by proving the following lemmas. After that, we
prove the convergence of Theorem.(\ref{converge_thm}). The used notations
are listed in Tab.\ref{tab:symbols}.

\begin{lemma}\label{bound_meta_grad}
	Let
	$
	C_3=N \beta C_2 (\sqrt{p}+\beta C_1)^{{N}-1}C_0 
	$
	and 
	$
	C_{L1}=2 C_1 + \mu C_3.
	$
	Based on the notations defined in Tab.\ref{tab:symbols},
	for arbitrary parameters 
	$\bm{z}_{w}, \bm{z}_{u} 
	\in \mathbb{R}^{p}$, 
	we have
	$$
	\| 
	\nabla_{\bm{z}_{w}}\mathcal{F} - 
	\nabla_{\bm{z}_{u}}\mathcal{F}
	\| \leq C_{L1} \| \bm{z}_{w} - \bm{z}_{u}\|.
	$$
\end{lemma}
\begin{proof}
	First, we have
	\begin{equation}
		\begin{aligned}
			\| \nabla_{\bm{z}} \mathcal{L}
			(\bm{\theta},\bm{\phi}) \| 
			  & \leq 
			\| \nabla_{\bm{\theta}} \mathcal{L}
			(\bm{\theta},\bm{\phi}) \| +
			\| \nabla_{\bm{\phi}} \mathcal{L}
			(\bm{\theta},\bm{\phi}) \|.\\
		\end{aligned}
	\end{equation}
	Based on Eq.(\ref{a_plus_b}), we can get that  
	\begin{equation}
		\begin{aligned}
			\nabla_{\bm{\theta}} \mathcal{L}
			(\bm{\theta},\bm{\phi}) & = 
			\nabla_{\bm{\theta}} L^{(in)}
			(\bm{\theta},\bm{\phi})\\
		\end{aligned}
	\end{equation}
	and 
	\begin{equation}\label{outer_grad}
		\begin{aligned}
			\nabla_{\bm{\phi}} \mathcal{L}
			(\bm{\theta},\bm{\phi}) & = 
			\nabla_{\bm{\phi}} L^{(in)}
			(\bm{\theta},\bm{\phi})+
			\mu \nabla_{\bm{\phi}} L^{(out)}
			(\bm{\theta}^{{(N)}},\bm{\phi}).\\
		\end{aligned}
	\end{equation}
	Thus we can derive that
	\begin{equation}\label{lemma_1_main}
		\begin{aligned}
			     & \| \nabla_{\bm{z}_{w}} \mathcal{L} 
			(\bm{\theta}_w, \bm{\phi}_w) - 
			\nabla_{\bm{z}_{u}} \mathcal{L}
			(\bm{\theta}_u, \bm{\phi}_u)
			\|\\
			\leq &                                        
			\| \nabla_{\bm{\theta}_w} L^{(in)}
			(\bm{\theta}_w,\bm{\phi}_w) -
			\nabla_{\bm{\theta}_u} L^{(in)}
			(\bm{\theta}_u,\bm{\phi}_u)
			\| \\
			     & + \| \nabla_{\bm{\phi}_w} L^{(in)}       
			(\bm{\theta}_w,\bm{\phi}_w) -
			\nabla_{\bm{\phi}_u} L^{(in)}
			(\bm{\theta}_u,\bm{\phi}_u)
			\| \\
			     & + \mu \| \nabla_{\bm{\phi}_w} L^{(out)}     
			({\bm{\theta}_{w}}^{{(N)}},\bm{\phi}_w) -
			\nabla_{\bm{\phi}_u} L^{(out)}
			({\bm{\theta}_{u}}^{{(N)}},\bm{\phi}_u)
			\| \\
			\leq & 2 C_1 \| \bm{z}_w - \bm{z}_u\| + \mu A_0,              
		\end{aligned}
	\end{equation}
	where we define
	\begin{equation}
		A_0=\| \nabla_{\bm{\phi}_w} L^{(out)}
		({\bm{\theta}_{w}}^{{(N)}},\bm{\phi}_w) -
		\nabla_{\bm{\phi}_u} L^{(out)}
		({\bm{\theta}_{u}}^{{(N)}},\bm{\phi}_u)
		\|.
	\end{equation}
	Notably, ${\bm{\theta}_{w}}^{{(N)}}$ and 
	${\bm{\theta}_{u}}^{{(N)}}$ are functions of $\bm{\phi}_w$ and $\bm{\phi}_u$,
	respectively.
				
	We next upper-bound $A$ in the above inequality. To this end, 
	we first regard 
	$L^{(out)}(\bm{\theta}^{{(N)}},\bm{\phi})$ as 
	a special case of ${N}$-step MAML loss \cite{MAML} where
	both $\bm{\theta}$ and $\bm{\phi}$ are meta-learned but
	with different learning rate. 
	Specifically, we reformulate the inner-level learning rate as a vector 
	$\bm{q}_\beta \in 
	\mathbb{R}^{p}$
	such that
	\begin{equation}
		\bm{q}_\beta[i]=
		\left\{
		\begin{array}{lr}
			\beta,                        
			\ \ i \leq {N_\theta}; \\
			0,                            
			\ \ otherwise.                \\
		\end{array}
		\right.
	\end{equation}
	This means that the inner learning rate for each parameter $i$ is 
	$\bm{q}_\beta[i]$. 
	Now the initial
	values $\bm{z}_{w}^{{(0)}}=\bm{z}$ are trainable
	parameters and $\bm{z}_{}^{{(i+1})}=
	[{(\bm{\theta}^{(i+1)}(\bm{\phi}))}^\top ,{\bm{\phi}}^\top]^\top
	=\bm{z}_{}^{{(i})}-\bm{q}_\beta \odot G_{w}^{(i,i)}$.
	Thus we can reformulate $A_0$ as following:
	\begin{equation}
		\begin{aligned}
			A_0 & =\| \nabla_{\bm{\phi}_w} L^{(out)}         
			({\bm{\theta}_{w}}^{{(N)}},\bm{\phi}_w) -
			\nabla_{\bm{\phi}_u} L^{(out)}
			({\bm{\theta}_{u}}^{{(N)}},\bm{\phi}_u)
			\| \\
			  & \leq \| \nabla_{\bm{z}} L^{(out)} 
			(\bm{z}_{w}^{{(N})}) -
			\nabla_{\bm{z}} L^{(out)}
			(\bm{z}_{u}^{{(N})}) \| \\
			  & =\|                                   
			\hat{G}_{w}^{(0,{N})}-
			\hat{G}_{u}^{(0,{N}))}
			\|.
		\end{aligned}
	\end{equation}
	Then we come to the inner-GD rollouts. For each
	$0 \leq j < {N}$, we have
	\begin{equation}
		\begin{aligned}
			     & \left\|              
			\bm{z}_{w}^{{(j+1})}-
			\bm{z}_{u}^{{(j+1})}\right\| \\
			=    & \left\|              
			\bm{z}_{w}^{{(j})}-
			\bm{z}_{u}^{{(j})}-
			\bm{q}_\beta \odot \left(
			G_{w}^{(j,j)}-
			G_{u}^{(j,j)}
			\right)
			\right\| \\
			\leq &                      
			\left\| 
			\bm{z}_{w}^{{(j})}-
			\bm{z}_{u}^{{(j})}
			\right\|+
			\left\|
			\bm{q}_\beta \odot \left(
			G_{w}^{(j,j)}-
			G_{u}^{(j,j)}
			\right)
			\right\| \\
			\leq &                      
			\left\| 
			\bm{z}_{w}^{{(j})}-
			\bm{z}_{u}^{{(j})}
			\right\|+
			\left\|
			\beta \left(
			G_{w}^{(j,j)}-
			G_{u}^{(j,j)}
			\right)
			\right\| \\
			=    &                      
			\left\| 
			\bm{z}_{w}^{{(j})}-
			\bm{z}_{u}^{{(j})}
			\right\|+\beta
			\left\|
			G_{w}^{(j,j)}-
			G_{u}^{(j,j)}
			\right\| \\
			=    & \left\|              
			\bm{z}_{w}^{{(j})}-
			\bm{z}_{u}^{{(j})}
			\right\|+\beta C_1
			\left\| 
			\bm{z}_{w}^{{(j})}-
			\bm{z}_{u}^{{(j})}
			\right\| \\
			=    & (1+\beta C_1)\left\| 
			\bm{z}_{w}^{{(j})}-
			\bm{z}_{u}^{{(j})}
			\right\| \\
			& \cdots \\
			\leq &        
			(1+\beta C_1)^{j+1}
			\left\| 
			\bm{z}_{w}^{{(0})}-
			\bm{z}_{u}^{{(0})}
			\right\| \\
			\leq &        
			(\sqrt{p}+\beta C_1)^{j+1}
			\left\| 
			\bm{z}_{w}-
			\bm{z}_{u}
			\right\|.
		\end{aligned}
	\end{equation}	
	From this, we can get the following inequality:
	\begin{equation}
		\begin{aligned}
			     & \left\|     
			\nabla_{\bm{z}_{w}^{{(N})}} L^{(out)}
			(\bm{z}_{w}^{{(N})}) -
			\nabla_{\bm{z}_{u}^{{(N})}} L^{(out)}
			(\bm{z}_{u}^{{(N})})
			\right\| \\
			\leq & C_1 \left\| 
			\bm{z}_{w}^{{(N})}-
			\bm{z}_{u}^{{(N})}
			\right\| \\
			\leq & C_1         
			(\sqrt{p}+ \beta C_1)^{{N}}
			\left\| 
			\bm{z}_{w}-
			\bm{z}_{u}
			\right\|.
		\end{aligned}
	\end{equation}
	Meanwhile, we have
	\begin{equation}\label{G_norm}
		\begin{aligned}
			\| 
			\hat{G}_{w}^{(j,{N})}  
			\| 
			=    &         
			\| 
			(
			\prod_{i=j}^{{N-1}}
			(
			\bm{I}-\bm{q_\beta}
			\odot H_{w}^{(i)} 
			)
			)
			\hat{G}_{w}^{({N},{N})}
			\| \\
			\leq & 
			(  
			\prod_{i=j}^{{N-1}}
			\Vert
			\bm{I}-\bm{q_\beta}
			\odot H_{w}^{(i)} 
			\Vert
			)
			\Vert
			\hat{G}_{w}^{({N},{N})}
			\Vert \\
			\leq & 
			(     
			\prod_{i=j}^{{N-1}}
			(
			\sqrt{p}+
			\Vert
			\beta H_{w}^{(i)} 
			\Vert
			))
			\Vert
			\hat{G}_{w}^{({N},{N})}
			\Vert \\
			\leq &         
			(\sqrt{p} + \beta C_1)^{{N}-j}C_0 
		\end{aligned}
	\end{equation}
	Furthermore, we introduce the following auxiliary notations for each
	$0 \leq i < {N}$ and $0 \leq j < {N}$: 
	\begin{equation}
		\begin{aligned}
			&A_1(j)=\hat{G}_{w}^{(j,N)}-\hat{G}_{u}^{(j,N)}, \\
			&A_2(j)=H_{w}^{(j)}\hat{G}_{w}^{(j+1,N)}-
			  H_{u}^{(j)}\hat{G}_{u}^{(j+1,N)},\\
			&A_3(j)=\left(
				H_{w}^{(j)}-
				H_{u}^{(j)}
				\right)
			 \hat{G}_{u}^{(j+1,{N})}, \\
			&A_4(j)= \beta C_2
			(\sqrt{p}+\beta C_1)^{{N}-j-1}C_0.
		\end{aligned}
	\end{equation}
	Based on these, we can get the following inequality:
	\begin{equation}
		\begin{aligned}
			&\Vert A_1(j) \Vert \\
			=      & 
			\left\| 
			A_1(j+1) -\bm{q}_\beta \odot A_2(j)
			\right\| \\
			=      &         
			\| 
			A_1(j+1)
			-\bm{q}_\beta \odot H_{w}^{(j)}A_1(j+1)
			-\bm{q}_\beta \odot A_3(j)
			\| \\
			\leq   &         
			\left\| 
			A_1(j+1)
			\right\| 
			+             
			\beta
			\| H_{w}^{(j)} \|
			\left\| 
			A_1(j+1)
			\right\| 
			+             
			\beta
			\left\| 
			A_3(j)
			\right\| 
			\\
			\leq   &         
			(\sqrt{p}+\beta C_1)
			\left\| 
			A_1(j+1)
			\right\| 
			+             
			A_4(j)
			\|
			\bm{z}_{w}^{{(j})}-
			\bm{z}_{u}^{{(j})}
			\|.
		\end{aligned}
	\end{equation}
	Thus the value of $A_0=\|A_1(0)\|$ can be obtained by applying 
	the following chain of inequalities applies 
	for each $0 \leq j < N$:
	\begin{equation}\label{bound_A}
		\begin{aligned}
			A_0
			\leq &
			(\sqrt{p}+\beta C_1)
			\left\| 
			A_1(1)
			\right\| 
			+             
			A_4(0)
			\|
			\bm{z}_{w}^{}-
			\bm{z}_{u}^{}
			\| \\
			\leq &
			(\sqrt{p}+\beta C_1)^2
			\left\| 
			A_1(2)
			\right\| \\
			& +  
			\sum_{j=0}^{1}           
			A_4(j) (\sqrt{p}+\beta C_1)^j
			\|
			\bm{z}_{w}^{}-
			\bm{z}_{u}^{}
			\|
			 \\
			& \cdots       \\
			\leq &
			(\sqrt{p}+\beta C_1)^N
			\left\| 
			A_1(N)
			\right\| \\
			& +  
			\sum_{j=0}^{N-1}           
			A_4(j) (\sqrt{p}+\beta C_1)^j
			\|
			\bm{z}_{w}^{}-
			\bm{z}_{u}^{}
			\| \\
			\leq      &              
			(\sqrt{p}+\beta C_1)^N C_1
			\left\| 
			\bm{z}_{w}-
			\bm{z}_{u}
			\right\| \\
			&              
			+        
			N \beta C_2 (\sqrt{p}+\beta C_1)^{{N}-1}C_0 
			\|
			\bm{z}_{w}^{}-
			\bm{z}_{u}^{}
			\|\\
			=         &              
			C_3
			\left\| 
			\bm{z}_{w}-
			\bm{z}_{u}
			\right\|.
		\end{aligned}
	\end{equation}
	Based on Eq.(\ref{bound_A}) and 
	Eq.(\ref{lemma_1_main}), we can get that
	\begin{equation}
		\begin{aligned}	
			     & \| 
			\nabla_{\bm{z}_{w}}\mathcal{F} - 
			\nabla_{\bm{z}_{u}}\mathcal{F}
			\| \\
			\leq &    
			\mathbb{E}(
			\| \nabla_{\bm{z}_{w}} \mathcal{L}
			(\bm{\theta}_w, \bm{\phi}_w) - 
			\nabla_{\bm{z}_{u}} \mathcal{L}
			(\bm{\theta}_u, \bm{\phi}_u)
			\|
			)\\
			\leq &    
			(2C_1 + \mu C_3)
			\left\| 
			\bm{z}_{w}-
			\bm{z}_{u}
			\right\| \\
		\end{aligned}
	\end{equation}
	which is equivalent to the statement of this Lemma.
\end{proof}

\begin{lemma}\label{bound_grad_norm}
	Let $C_{L2}=C^2_0 + \mu (\sqrt{p}+\beta C_1)^{{2N}}C_{0}^2$.
	Based on the notations defined in Tab.\ref{tab:symbols},
	for arbitrary meta-model with parameters
	$\bm{\phi} \in \mathbb{R}^{{N_\phi}}$,
	the outer-level gradient $\mathcal{G}_{\bm{\phi}}$ holds that
	$\mathbb{E}\left(
	\|
	\mathcal{G}_{\bm{\phi}}
	\|^2
	\right) \leq C_{L2}$.
\end{lemma}

\begin{proof}
	According to Eq.(\ref{outer_grad}), we can derive that
	\begin{equation}
		\begin{aligned}
			\| \mathcal{G}_{\bm{\phi}} \|^2 
			& \leq 
			\| \nabla_{\bm{\phi}} L^{(in)} (\bm{\theta},\bm{\phi})\|^2  +
			\mu \|  \nabla_{\bm{\phi}} 
			   L^{(out)}(\bm{\theta}^{{(N)}},\bm{\phi}) \|^2 
			\\
			& \leq 
			\| \nabla_{\bm{z}} L^{(in)} (\bm{z})\|^2  +
			\mu \|  \nabla_{\bm{z}} 
			   L^{(out)}(\bm{z}^{{(N)}}) \|^2 
			\\
			& \leq 
			C^2_0 + \mu \| \hat{G}_{w}^{(0,{N})}  \|^2 
			\\
		\end{aligned}
	\end{equation}
	Based on Eq.(\ref{G_norm}), we have
	\begin{equation}
		\begin{aligned}
			\| 
			\hat{G}_{w}^{(0,{N})}  
			\|^2 
			\leq &         
			((\sqrt{p} + \beta C_1)^{{N}-j}C_0)^2 \\
			= &
			(\sqrt{p}+\beta C_1)^{{2N}}C_{0}^2.
		\end{aligned}
	\end{equation}
	Now we can get that
	\begin{equation}
		\begin{aligned}
			\| \mathcal{G}_{\bm{\phi}} \| 
			& \leq 
			C^2_0 + \mu \| \hat{G}_{w}^{(0,{N})}  \|
			\\
			& \leq 
			C^2_0 + \mu (\sqrt{p}+\beta C_1)^{{2N}}C_{0}^2.
		\end{aligned}
	\end{equation}
	Thus the proof is is concluded.
\end{proof}

\begin{lemma}\label{bound_grad_M_2}
	Let
	$
	C_{L3}=\frac{C_{L1}}{2} (C^2_0 + C_{L2})
	$.
	Based on Lemma.(\ref{bound_meta_grad}) and Lemma.(\ref{bound_grad_norm}),
	as well as notations defined in Tab.\ref{tab:symbols},
	we have the following inequality for each $k \in \mathbb{N}$:
	$
	\sum_{i=1}^{k} \gamma_{i} 
				\|\nabla_{\bm{z}}\mathcal{F}_{i-1}\|^2
		\leq
		\mathcal{F}_{0} + C_{L3}\sum_{i=1}^{k}\gamma_{i}^2.
	$
\end{lemma}

\begin{proof}
	For each $0 < i \leq k$, we have
	\begin{equation}\label{F_z_theta_phi}
		\begin{aligned}
			\|\nabla_{\bm{z}}\mathcal{F}_{i-1}\|^2       
			\leq                                             
			\|\nabla_{\bm{\theta}}\mathcal{F}_{i-1}\|^2+ 
			\|\nabla_{\bm{\phi}}\mathcal{F}_{i-1}\|^2.   
		\end{aligned}
	\end{equation}
	So we can derive that
	\begin{equation}\label{M_two_terms}
		\begin{aligned}
			     & \sum_{i=1}^{k} \gamma_{i} 
			\|\nabla_{\bm{z}}\mathcal{F}_{i-1}\|^2
			\\
			\leq &                                       
			\sum_{i=1}^{k} \gamma_{i} (
			\|\nabla_{\bm{\theta}}\mathcal{F}_{i-1}\|^2
			+
			\|\nabla_{\bm{\phi}}\mathcal{F}_{i-1}\|^2
			) \\
			\leq &                                       
			\sum_{i=1}^{k} \gamma_{i} 
			\|\nabla_{\bm{\theta}}\mathcal{F}_{i-1}\|^2
			+
			\sum_{i=1}^{k} \gamma_{i} 
			\|\nabla_{\bm{\phi}}\mathcal{F}_{i-1}\|^2
			.
		\end{aligned}
	\end{equation}	
	Based on Eq.(4.3) from \cite{SGD}, 
	we have
	\begin{equation}\label{term_1_M}
		\begin{aligned}
			\mathcal{F}_i
			\leq &   
			\mathcal{F}_{i-1} +
			\nabla_{\bm{\theta}}\mathcal{F}_{i-1}^\top
			( \bm{\theta}_{(i)}-\bm{\theta}_{(i-1)} )
			+ 
			\frac{C_{L1}}{2} \| \bm{\theta}_{(i)}-\bm{\theta}_{(i-1)} \|^2 \\
			& +
			\nabla_{\bm{\phi}}\mathcal{F}_{i-1}^\top
			( \bm{\phi}_{(i)}-\bm{\phi}_{(i-1)} )
			+ 
			\frac{C_{L1}}{2} \| \bm{\phi}_{(i)}-\bm{\phi}_{(i-1)} \|^2 \\
			\leq &   
			\mathcal{F}_{i-1} -
			\gamma_i
			\nabla_{\bm{\theta}}\mathcal{F}_{i-1}^\top
			\mathcal{G}_{\bm{\theta}}+
			\frac{\gamma_{i}^2C_{L1}}{2} 
			\|\mathcal{G}_{\bm{\theta}}\|^2 \\
			& -
			\gamma_i
			\nabla_{\bm{\phi}}\mathcal{F}_{i-1}^\top
			\mathcal{G}_{\bm{\phi}}+
			\frac{\gamma_{i}^2C_{L1}}{2} 
			\|\mathcal{G}_{\bm{\phi}}\|^2.
		\end{aligned}
	\end{equation}
	Taking the expectation over $\mathcal{P}(\mathcal{T})$ 
	on both sides of the equation, 
	we can get that 
	\begin{equation}\label{term_1_M_E}
		\begin{aligned}
			\mathcal{F}_i
			\leq &   
			\mathcal{F}_{i-1} -
			\gamma_i
			\nabla_{\bm{\theta}}\mathcal{F}_{i-1}^\top
			\mathbb{E}(\mathcal{G}_{\bm{\theta}})+
			\frac{\gamma_{i}^2 C_{L1}}{2} 
			\mathbb{E}(\|\mathcal{G}_{\bm{\theta}}\|^2) \\
			& -
			\gamma_i
			\nabla_{\bm{\phi}}\mathcal{F}_{i-1}^\top
			\mathbb{E}(\mathcal{G}_{\bm{\phi}})+
			\frac{\gamma_{i}^2 C_{L1}}{2} 
			\mathbb{E}(\|\mathcal{G}_{\bm{\phi}}\|^2) \\
			\leq &   
			\mathcal{F}_{i-1} -
			\gamma_i
			\| \nabla_{\bm{\theta}}\mathcal{F}_{i-1} \|^2 + 
			\frac{\gamma_{i}^2 C_{L1}}{2} 
			\mathbb{E}(\|\mathcal{G}_{\bm{\theta}}\|^2) \\
			& -
			\gamma_i
			\| \nabla_{\bm{\phi}}\mathcal{F}_{i-1} \|^2 + 
			\frac{\gamma_{i}^2 C_{L1}}{2} 
			\mathbb{E}(\|\mathcal{G}_{\bm{\phi}}\|^2) \\
			\leq &   
			\mathcal{F}_{i-1} -
			\gamma_i
			\| \nabla_{\bm{\theta}}\mathcal{F}_{i-1} \|^2 + 
			\frac{\gamma_{i}^2 C_{L1}}{2} C^2_0 \\
			& -
			\gamma_i
			\| \nabla_{\bm{\phi}}\mathcal{F}_{i-1} \|^2 + 
			\frac{\gamma_{i}^2 C_{L1}}{2} C_{L2}. \\
		\end{aligned}
	\end{equation}
	This inequality is equivalent to
	\begin{equation}\label{theta_i}
		\begin{aligned}
			& \gamma_i
			\| \nabla_{\bm{\theta}}\mathcal{F}_{i-1} \|^2
			+ 
			\gamma_i
			\| \nabla_{\bm{\phi}}\mathcal{F}_{i-1} \|^2 \\
			\leq &   
			\mathcal{F}_{i-1} -
			\mathcal{F}_i + 
			\frac{\gamma_{i}^2 C_{L1}}{2} (C^2_0 + C_{L2}). \\
		\end{aligned}
	\end{equation}
	Summing up Eq.(\ref{theta_i}) for all $0 < i \leq k$ yields that
	\begin{equation}\label{bound_term_1}
		\begin{aligned}
			&\sum_{i=1}^{k}
			\gamma_i
			\| \nabla_{\bm{\theta}}\mathcal{F}_{i-1} \|^2
			+ 
			\sum_{i=1}^{k}
			\gamma_i
			\| \nabla_{\bm{\phi}}\mathcal{F}_{i-1} \|^2 \\
			\leq &   
			\mathcal{F}_{0} -
			\mathcal{F}_{k} + 
			\sum_{i=1}^{k}
			\frac{\gamma_{i}^2 C_{L1}}{2} (C^2_0 + C_{L2}). \\
			\leq &   
			\mathcal{F}_{0} + 
			\sum_{i=1}^{k}
			\frac{\gamma_{i}^2 C_{L1}}{2} (C^2_0 + C_{L2}). \\
		\end{aligned}
	\end{equation}
	Thus we have
	\begin{equation}\label{M_gamma}
		\begin{aligned}
			     & \sum_{i=1}^{k} \gamma_{i} 
			\|\nabla_{\bm{z}}\mathcal{F}_{i-1}\|^2
			 \\
			\leq &                                       
			\sum_{i=1}^{k} \gamma_{i} 
			\|\nabla_{\bm{\theta}}\mathcal{F}_{i-1}\|^2
			+
			\sum_{i=1}^{k} \gamma_{i}
			\|\nabla_{\bm{\phi}}\mathcal{F}_{i-1}\|^2
			 \\
			\leq &                                       
			\mathcal{F}_{0} + 
			\sum_{i=1}^{k}
			\frac{\gamma_{i}^2 C_{L1}}{2} (C^2_0 + C_{L2}). \\
		\end{aligned}
	\end{equation}	
	In this way, the proof is concluded.
\end{proof}

\begin{proof}[Proof of Theorem.(\ref{converge_thm})]$ $\\
	\textbf{Proof of Statement (1).}
	Based on Eq.(\ref{term_1_M}) and the result of
	Lemma.(\ref{bound_grad_norm}), we can
	get the following inequality by algebraic manipulation:
	\begin{equation}\label{FF}
		\begin{aligned}
			\mathcal{F}_{i-1} - \mathcal{F}_i 
			> &   
			\gamma_i
			C^2_0
			+ 
			\gamma_i
			C_{L2}
			-
			\frac{\gamma_{i}^2}{2} C_{L3}.
		\end{aligned}
	\end{equation}
	From this, we can derive that when $\gamma_i < 2 \frac{C^2_0+C_{L2}}{C_{L3}}$,
	the right-side of Eq.(\ref{FF}) would be positive, thereby ensuring
	$\mathcal{F}_{i-1} > \mathcal{F}_i$. Thus the statement 1 is proved.\\

	\noindent
	\textbf{Proof of Statement (2).}
	Since we have the condition that $\sum_{i=1}^{\infty}\gamma_{i}^2<\infty$, 
	the right-hand side of the statement of Lemma.(\ref{bound_grad_M_2}) 
	converges to a finite value 
	when $k \rightarrow \infty$. Thus the left-hand side is also convergent. 
	Now suppose this statement is false. Then there exists $k_0 \in \mathbb{N}$
	and $B>0$ such that 
	$\forall i \geq k_0:\|\nabla_{\bm{z}}\mathcal{F}_{i-1}\|^2>B$. 
	However, for all $i > k_0$ we have
	\begin{equation}
		\sum_{i=1}^{k} 
		\gamma_{i} 
		\|\nabla_{\bm{z}}\mathcal{F}_{i-1}\|^2
		\geq
		B \sum_{i=k_0}^{k} \gamma_{i} 
		\rightarrow \infty
	\end{equation}
	while $k \rightarrow \infty$, which is a contradiction.
	Thus we can get that
	\begin{equation}\label{F_infty}
		\lim_{i \rightarrow \infty}
		{\|\nabla_{\bm{z}}\mathcal{F}_{i-1}\|^2}=0.
	\end{equation}
	Taking a step further, the parameters $\bm{z}$ is bounded as
	$\| \bm{z}_{(i)} \| < C_r$.  According to the Bolzano-Weierstrass 
	theorem, any bounded sequence must have
	convergent subsequence, which immediately suggests the existence of at 
	least one limit point for $\{\bm{z}_{(i)}\}_i$.
	Picking an arbitrary convergent subsequence $\{\bm{z}_{(s_i)}\}_i$ with a 
	limit point $\bm{z}^\star$. Since $s_i \rightarrow \infty$ when 
	$i \rightarrow \infty$, we can get that
	\begin{equation}\label{F_S_i}
		\lim_{i \rightarrow \infty}
		{\|\nabla_{\bm{z}}\mathcal{F}(\bm{z}_{(s_i)})\|^2} = 0
	\end{equation}
	following the similar spirit of Eq.(\ref{F_infty}). 
	Using the continuity of $\nabla_{\bm{z}}\mathcal{F}$, we
	have $\nabla_{\bm{z}}\mathcal{F}(\bm{z}_{(s_i)}) \rightarrow 
	\nabla_{\bm{z}}\mathcal{F}(\bm{z}^\star)$
	when $i \rightarrow \infty$.
	Thus Eq.(\ref{F_S_i}) is equivalent to
	$
		{\|\nabla_{\bm{z}}\mathcal{F}(\bm{z}^\star)\|^2} = 0
	$, which means that $\{\bm{z}_{(s_i)}\}_i$ converges to 
	a stationary point. \\

	\noindent
	\textbf{Proof of Statement (3).} We have
	\begin{equation}
		\begin{aligned}
			\min_{0 \leq j < k} \left( 
			\|\nabla_{\bm{z}}\mathcal{F}_{j}\|^2
			\right)
			\sum_{i=1}^{k} \gamma_{i} 
			\leq &                            
			\sum_{i=1}^{k} \gamma_{i}
			\|\nabla_{\bm{z}}\mathcal{F}_{i-1}\|^2 \\
			\leq &                            
			\mathcal{F}_{0} + 
			C_{L3}
			\sum_{i=1}^{k}\gamma_{i}^2.
		\end{aligned}
	\end{equation}
	Dividing $\sum_{i=1}^{k} \gamma_{i}$ yields that
	\begin{equation}
		\begin{aligned}
			\min_{0 \leq j < k}  
			\|\nabla_{\bm{z}}\mathcal{F}_{j}\|^2 
			\leq &                            
			\frac{\mathcal{F}_{0}}{\sum_{i=1}^{k} \gamma_{i}} + 
			C_{L3}
			\frac{\sum_{i=1}^{k}\gamma_{i}^2}{\sum_{i=1}^{k} \gamma_{i}}.
		\end{aligned}
	\end{equation}
	In this way, statement 3 is satisfied by observing that
	\begin{equation}
		\begin{aligned}
			  & \sum_{i=1}^{k}\gamma_{i}=\sum_{i=1}^{k}i^{-0.5}=\bm{\Omega}(k^{0.5}), \\
			  & \sum_{i=1}^{k}\gamma_{i}^2=\sum_{i=1}^{k}i^{-1}=O(\log k)=\bm{\Omega}(k^\epsilon)   
		\end{aligned}
	\end{equation}
	for any $\epsilon > 0$.
\end{proof}

\end{document}